\documentclass[11pt]{article}
\usepackage{makeidx}
\usepackage{euscript}
\usepackage{amsmath,amssymb,amsfonts}
\usepackage{dsfont}
\usepackage{latexsym}
\usepackage{graphicx}
\usepackage{fancybox}
\usepackage{fullpage}
\usepackage[backref]{hyperref}
\usepackage{paralist}
\usepackage{color}
\usepackage{xcolor}
\usepackage{wrapfig}
\usepackage{tikz}
\usepackage{setspace}
\usepackage{algorithm}
\usepackage[noend]{algpseudocode}
\usepackage[framemethod=tikz]{mdframed}
\usepackage{xspace}
\usepackage{pgfplots}
\pgfplotsset{compat=1.5}

\newenvironment{proof}{\noindent{\bf Proof : \ }}{\hfill$\Box$\par\medskip}

\newtheorem{theorem}{Theorem}
\newtheorem{corollary}[theorem]{Corollary}
\newtheorem{lemma}[theorem]{Lemma}

\newtheorem{claim}[theorem]{\bf Claim}

\newtheorem{observation}[theorem]{\bf Observation}
\newtheorem{example}[theorem]{Example}

\newenvironment{proofof}[1]{\begin{trivlist} \item {\bf Proof
#1:~~}}
  {\qed\end{trivlist}}
\renewenvironment{proofof}[1]{\par\medskip\noindent{\bf Proof of #1: \ }}{\hfill$\Box$\par\medskip}
\newcommand{\namedref}[2]{\hyperref[#2]{#1~\ref*{#2}}}
\newcommand{\thmlab}[1]{\label{thm:#1}}
\newcommand{\thmref}[1]{\namedref{Theorem}{thm:#1}}
\newcommand{\obslab}[1]{\label{obs:#1}}
\newcommand{\obsref}[1]{\namedref{Observation}{obs:#1}}
\newcommand{\lemlab}[1]{\label{lem:#1}}
\newcommand{\lemref}[1]{\namedref{Lemma}{lem:#1}}
\newcommand{\claimlab}[1]{\label{claim:#1}}
\newcommand{\claimref}[1]{\namedref{Claim}{claim:#1}}
\newcommand{\corlab}[1]{\label{cor:#1}}
\newcommand{\corref}[1]{\namedref{Corollary}{cor:#1}}
\newcommand{\seclab}[1]{\label{sec:#1}}
\newcommand{\secref}[1]{\namedref{Section}{sec:#1}}

\newcommand{\figlab}[1]{\label{fig:#1}}
\newcommand{\figref}[1]{\namedref{Figure}{fig:#1}}

\newenvironment{remindertheorem}[1]{\medskip \noindent {\bf Reminder of  #1.  }\em}{}

\newcommand{\COMMENTED}[1]{{}}

\newcommand{\HAM}[1]{\ensuremath{\mathsf{HAM}\left(#1\right)}}
\renewcommand{\gcd}[1]{\ensuremath{\mathsf{gcd}\left(#1\right)}}

\newcommand{\IND}{\ensuremath{\mathsf{IND}}}

\newcommand{\mdef}[1]{{\ensuremath{#1}}\xspace}  

\newcommand{\superscript}[1]{\ensuremath{^{\mbox{\tiny{\textit{#1}}}}}\xspace}
\def \th {\superscript{th}}     
\def \etal{{\it et~al.}}

\def \poly {\mdef{\text{poly}}}   
\renewcommand{\O}[1]{\ensuremath{\mathcal{O}\left(#1\right)}}						

\newcommand{\ignore}[1]{}

\newif\ifnotes\notestrue 
\ifnotes
\newcommand{\elena}[1]{\textcolor{red}{{\bf (Elena:} {#1}{\bf ) }} \marginpar{\tiny\bf
             \begin{minipage}[t]{0.5in}
               \raggedright E:
                \end{minipage}}}
\newcommand{\erfan}[1]{\textcolor{blue}{{\bf (Erfan:} {#1}{\bf ) }} \marginpar{\tiny\bf
             \begin{minipage}[t]{0.5in}
               \raggedright E:
                \end{minipage}}}
\newcommand{\samson}[1]{\textcolor{green}{{\bf (Samson:} {#1}{\bf ) }} \marginpar{\tiny\bf
             \begin{minipage}[t]{0.5in}
               \raggedright S:
            \end{minipage}}}
\newcommand{\funda}[1]{\textcolor{purple}{{\bf (Funda:} {#1}{\bf ) }} \marginpar{\tiny\bf
             \begin{minipage}[t]{0.5in}
               \raggedright F:
            \end{minipage}}}            																										
\else
\newcommand{\elena}[1]{}
\newcommand{\erfan}[1]{}
\newcommand{\samson}[1]{}
\newcommand{\funda}[1]{}
\fi

\definecolor{mahogany}{rgb}{0.75, 0.25, 0.0}
\definecolor{darkblue}{rgb}{0.0, 0.0, 0.55}
\definecolor{darkpastelgreen}{rgb}{0.01, 0.75, 0.24}
\definecolor{darkgreen}{rgb}{0.0, 0.2, 0.13}
\definecolor{darkgoldenrod}{rgb}{0.72, 0.53, 0.04}
\definecolor{darkred}{rgb}{0.55, 0.0, 0.0}
\definecolor{forestgreen}{rgb}{0.13, 0.55, 0.13}
\hypersetup{
     colorlinks   = true,
     citecolor    = mahogany,
		 linkcolor		= forestgreen,
		 urlcolor			= forestgreen
}

\makeatletter
\renewcommand*{\@fnsymbol}[1]{\textcolor{mahogany}{\ensuremath{\ifcase#1\or *\or \dagger\or \ddagger\or
 \mathsection\or \triangledown\or \mathparagraph\or \|\or **\or \dagger\dagger
   \or \ddagger\ddagger \else\@ctrerr\fi}}}
\makeatother

\providecommand{\email}[1]{\href{mailto:#1}{\nolinkurl{#1}\xspace}}
\title{Streaming Periodicity with Mismatches\footnote{
A preliminary version of this paper is to appear in the Proceedings of the 21st International Workshop on Randomization and Computation (RANDOM 2017)
}}
\author{
Funda Erg{\"{u}}n\thanks{School of Informatics and Computing, Indiana University, Bloomington, IN.
Research supported by NSF CCF-1619081.
Email: \email{fergun@indiana.edu}.}
\and
Elena Grigorescu\thanks{Department of Computer Science, Purdue University, West Lafayette, IN. 
Research supported by NSF CCF-1649515.
Email: \email{elena-g@purdue.edu}.}
\and
Erfan Sadeqi Azer\thanks{School of Informatics and Computing, Indiana University, Bloomington, IN.
Email: \email{esadeqia@indiana.edu}.}
\and
Samson Zhou\thanks{Department of Computer Science, Purdue University, West Lafayette, IN. 
Research supported by NSF CCF-1649515. 
Email: \email{samsonzhou@gmail.com}.}
}

\begin{document}
\maketitle
\newcommand{\figtchj}{
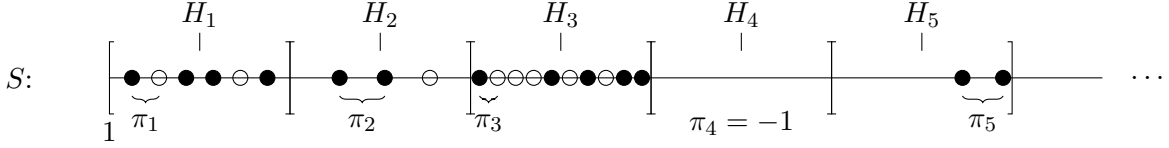
\begin{figure*}[htb]
\centering
\begin{tikzpicture}[scale=0.6]

\draw (-10,0) -- (12,0);
\node at (13,0){$\ldots$};
\node at (-12,0){$S$:};
\node at (-10,-1.2){$1$};
\foreach \x in {-10,-6,...,6}{
	\draw (\x+0.1,0.8) -- (\x,0.8) -- (\x,-0.8) -- (\x+0.1,-0.8);
	\draw (\x+3.9,0.8) -- (\x+4,0.8) -- (\x+4,-0.8) -- (\x+3.9,-0.8);
	\draw (\x+2,0.6) -- (\x+2,1);
}
\node at (-8,1.4){$H_1$};
\node at (-4,1.4){$H_2$};
\node at (0,1.4){$H_3$};
\node at (4,1.4){$H_4$};
\node at (8,1.4){$H_5$};

\node[draw,circle,inner sep=2pt,fill] at (-9.5,0) {};
\draw[decorate,decoration={brace,mirror}](-9.5,-0.4cm) -- (-8.9cm,-0.4cm);
\node at (-9.2,-1){$\pi_1$};
\node[draw,circle,inner sep=2pt] at (-8.9,0) {};
\node[draw,circle,inner sep=2pt,fill] at (-8.3,0) {};
\node[draw,circle,inner sep=2pt,fill] at (-7.7,0) {};
\node[draw,circle,inner sep=2pt] at (-7.1,0) {};
\node[draw,circle,inner sep=2pt,fill] at (-6.5,0) {};

\node[draw,circle,inner sep=2pt,fill] at (-4.9,0) {};
\draw[decorate,decoration={brace,mirror}](-4.9,-0.4cm) -- (-3.9cm,-0.4cm);
\node at (-4.4,-1){$\pi_2$};
\node[draw,circle,inner sep=2pt,fill] at (-3.9,0) {};
\node[draw,circle,inner sep=2pt] at (-2.9,0) {};

\node[draw,circle,inner sep=2pt,fill] at (-1.8,0) {};
\draw[decorate,decoration={brace,mirror}](-1.8,-0.4cm) -- (-1.4cm,-0.4cm);
\node at (-1.6,-1){$\pi_3$};
\node[draw,circle,inner sep=2pt] at (-1.4,0) {};
\node[draw,circle,inner sep=2pt] at (-1,0) {};
\node[draw,circle,inner sep=2pt] at (-0.6,0) {};
\node[draw,circle,inner sep=2pt,fill] at (-0.2,0) {};
\node[draw,circle,inner sep=2pt] at (0.2,0) {};
\node[draw,circle,inner sep=2pt,fill] at (0.6,0) {};
\node[draw,circle,inner sep=2pt] at (1,0) {};
\node[draw,circle,inner sep=2pt,fill] at (1.4,0) {};
\node[draw,circle,inner sep=2pt,fill] at (1.8,0) {};

\node at (4,-1){$\pi_4=-1$};

\node[draw,circle,inner sep=2pt,fill] at (8.9,0) {};
\draw[decorate,decoration={brace,mirror}](8.9,-0.4cm) -- (9.8cm,-0.4cm);
\node at (9.35,-1){$\pi_5$};
\node[draw,circle,inner sep=2pt,fill] at (9.8,0) {};

\end{tikzpicture}
\caption{Observe that all dots in each interval are equally spaced after the first. These dots represent $\mathcal{T}^c$: the black dots represent $\mathcal{T}$, while the white dots are added to convert the irregularly spaced black dots into regularly spaced dot sequences.}\figlab{fig:tchj}
\end{figure*}
}

\newcommand{\figcomparison}{
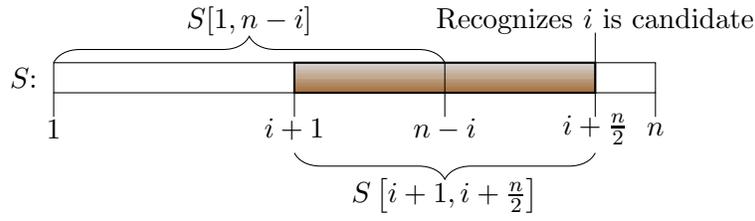
\begin{figure*}[htb]
\centering
\begin{tikzpicture}[scale=0.4]

\draw (0cm,0cm) rectangle+(8cm,1cm);
\filldraw[thick, top color=white,bottom color=brown] (8cm,0cm) rectangle+(10cm,1cm);
\draw (18cm,0cm) rectangle+(2cm,1cm);
\draw [decorate,decoration={brace,mirror,amplitude=10pt}](8cm,-2cm) -- (18cm,-2cm);
\node at (13cm, -3.4cm){$S\left[i+1,i+\frac{n}{2}\right]$};
\node at (-1cm, 0.5cm){$S$:};

\draw [decorate,decoration={brace,mirror,amplitude=10pt}](13cm,1cm) -- (0cm,1cm);
\node at (6.5cm, 2.4cm){$S[1,n-i]$};

\draw (0cm,0cm) -- (0cm,-0.8cm);
\node at (0cm, -1.2cm){$1$};

\draw (8cm,0cm) -- (8cm,-0.8cm);
\node at (8cm, -1.2cm){$i+1$};

\draw (13cm,1cm) -- (13cm,-0.8cm);
\node at (13cm, -1.2cm){$n-i$};

\draw (18cm,0cm) -- (18cm,-0.8cm);
\node at (18cm, -1.2cm){$i+\frac{n}{2}$};

\draw (18cm,1cm) -- (18cm,2cm);
\node at (18cm, 2.4cm){Recognizes $i$ is candidate};

\draw (20cm,0cm) -- (20cm,-0.8cm);
\node at (20cm, -1.2cm){$n$};

\end{tikzpicture}
\caption{When $i$ is recognized as a candidate, the algorithm has already passed $n-i$ and cannot build $S[1,n-i]$.}\figlab{fig:comparison}
\end{figure*}
}

\newcommand{\figlattice}{
\begin{figure*}[htb]
\centering
\begin{tikzpicture}[scale=0.7]
    \foreach \x in {-4,-3}{
      \foreach \y in {-4.8,-3.6,...,4.8}{
        \node[draw,circle,inner sep=2pt,fill] at (\x,\y) {};
      }
    }
		\foreach \x in {4}{
      \foreach \y in {-4.8,-3.6,...,4.8}{
        \node[draw,circle,inner sep=2pt,fill] at (\x,\y) {};
      }
    }
		\foreach \x in {-3,-2,...,3}{
      \foreach \y in {3.6,2.4,...,-3.6}{
        \node[draw,circle,inner sep=2pt] at (\x,\y) {};
      }
			\foreach \y in {-3.6,-4.8}{
			        \node[draw,circle,inner sep=2pt,fill] at (\x,\y) {};
      }
			\foreach \y in {4.8}{
			        \node[draw,circle,inner sep=2pt,fill] at (\x,\y) {};
      }
    }
\node at (-0.05,0)[below right]{\tiny{$i$}};
\node at (0.95,0)[below right]{\tiny{$i+q$}};
\node at (-1.05,0)[below right]{\tiny{$i-q$}};
\node at (-0.05,1.2)[below right]{\tiny{$i+p$}};
\node at (-0.05,-1.2)[below right]{\tiny{$i-p$}};

\foreach \x in {-2,-1,...,3}{
	\draw (\x,3.75)[line width = 1.5] -- (\x,4.8);
	\draw (\x,-2.55)[line width = 1.5] -- (\x,-3.6);
	\foreach \y in {3.6, 2.4,...,-2.4}{
		\draw(\x,\y-0.2)[dashed] -- (\x,\y-1);
	}
	\draw(\x,-3.8)[dashed] -- (\x,-4.6);
}
\foreach \x in {-4,-3}{
	\foreach \y in {4.8, 3.6,...,-3.6}{
		\draw(\x,\y-0.2)[dashed] -- (\x,\y-1);
	}
}
\foreach \y in {4.8, 3.6,...,-3.6}{
		\draw(4,\y-0.2)[dashed] -- (4,\y-1);
}

\foreach \y in {-2.4,-1.2,...,3.6}{
	\draw (-3,\y)[line width = 1.5] -- (-2.15,\y);
	\draw (3.15,\y)[line width = 1.5] -- (4,\y);
	\foreach \x in {-3, -2,...,3}{
		\draw(\x+0.2,\y)[dashed] -- (\x+0.8,\y);
	}
	\draw(-3.8,\y)[dashed] -- (-3.2,\y);
}

\foreach \y in {-4.8,-3.6}{
	\foreach \x in {3,2,...,-4}{
		\draw(\x+0.2,\y)[dashed] -- (\x+0.8,\y);
	}
}
\foreach \x in {3,2,...,-4}{
		\draw(\x+0.2,4.8)[dashed] -- (\x+0.8,4.8);
}

\end{tikzpicture}
\caption{The dashed lines are good edges and the solid lines are bad edges. Note that it is impossible to go from an isolated (light) node to one outside the the enclosed region (i.e., to a dark node) without traversing through a bad edge. The total number of enclosed edges can be at most $k^2$ if the number of bad edges is at most $4k$.}\figlab{fig:lattice}
\end{figure*}
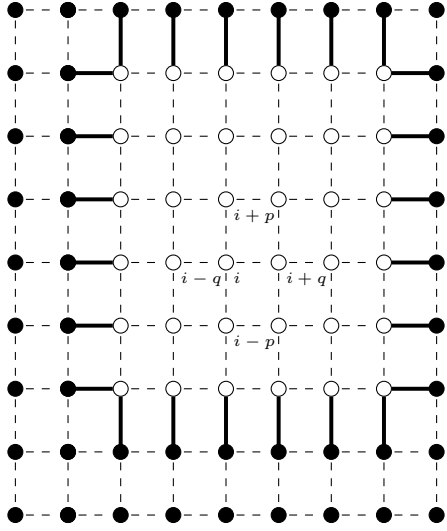
}

\newcommand{\figsidelattice}{
\begin{figure*}[htb]
\centering
\begin{tikzpicture}[scale=0.6]
    \foreach \x in {2,3,4}{
      \foreach \y in {-4.8,-3.6,...,4.8}{
        \node[draw,circle,inner sep=2pt,fill] at (\x,\y) {};
      }
    }
		\foreach \x in {-4,-3,...,1}{
      \foreach \y in {-4.8,-3.6,...,2.4}{
        \node[draw,circle,inner sep=2pt] at (\x,\y) {};
      }
			\foreach \y in {3.6,4.8}{
        \node[draw,circle,inner sep=2pt,fill] at (\x,\y) {};
      }
    }
\node at (-4,-5)[below]{\tiny{$i$}};
\node at (-3,-5)[below]{\tiny{$i+q$}};
\node at (-1.8,-5)[below]{\tiny{$i+2q$}};
\node at (-4.2,-3.6)[left]{\tiny{$i+p$}};
\node at (-4.2,-2.4)[left]{\tiny{$i+2p$}};

\foreach \x in {-4,-3,...,1}{
 \draw (\x,3.6)[dashed] -- (\x,4.8);
 \draw (\x,2.57)[line width=1.2] -- (\x,3.6);
	\foreach \y in {3.6,2.4,...,-4.8}{
	\draw (\x,\y-0.2)[dashed] -- (\x,\y-1);
	}
}
\foreach \x in {2,3,4}{
	\foreach \y in {4.8,3.6,...,-3.6}{
	\draw (\x,\y-0.2)[dashed] -- (\x,\y-1);
	}
}

\foreach \y in {-4.8,-3.6,...,1.2,2.4}{
 \draw (3,\y)[dashed] -- (4,\y);
 \draw (2,\y)[dashed] -- (3,\y);
 \draw (1.17,\y)[line width=1.2] -- (2,\y);
	\foreach \x in {-4,-3,...,0}{
	\draw (\x+0.2,\y)[dashed] -- (\x+0.8,\y);
	}
	}
	
\foreach \y in {3.6,4.8}{
	\foreach \x in {-4,-3,...,3}{
	\draw (\x,\y)[dashed] -- (\x+1,\y);
	}
}

\end{tikzpicture}
\caption{The dashed lines are good edges and the solid lines are bad edges. Part of the boundary of the enclosed points is induced by the boundary of the grid. The total area of the enclosed regions is at most $k^2$ if the perimeter of the bad edges is at most $2k$.}\figlab{fig:side:lattice}
\end{figure*}
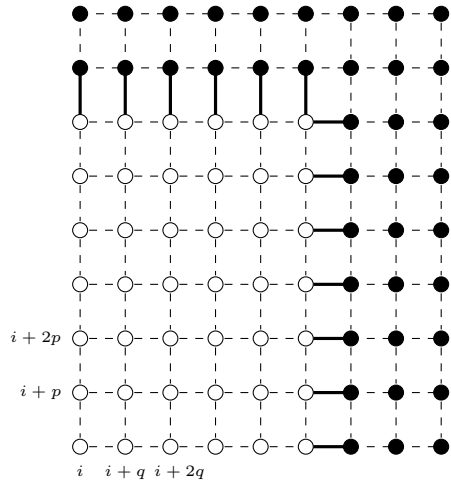
}

\newcommand{\figcube}{
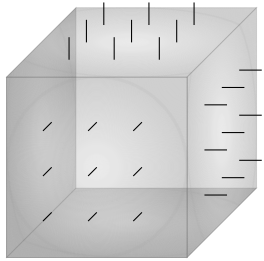
\begin{figure*}[htb]
\centering
\begin{tikzpicture}[scale=0.6]
\filldraw[shading=radial,inner color=white, outer color=gray!75, opacity=0.1](0,0,0) -- (4,0,0) -- (4,0,4) -- (0,0,4) -- (0,0,0);
\filldraw[shading=radial,inner color=white, outer color=gray!75, opacity=0.1](0,0,0) -- (0,4,0) -- (0,4,4) -- (0,0,4) -- (0,0,0);
\filldraw[shading=radial,inner color=white, outer color=gray!75, opacity=0.1](0,0,4) -- (4,0,4) -- (4,4,4) -- (0,4,4) -- (0,0,4);
\filldraw[shading=radial,inner color=white, outer color=gray!75, opacity=0.1](0,4,0) -- (4,4,0) -- (4,4,4) -- (0,4,4) -- (0,4,0);
\filldraw[shading=radial,inner color=white, outer color=gray!75, opacity=0.1](4,0,0) -- (4,4,0) -- (4,4,4) -- (4,0,4) -- (4,0,0);

\draw (4,1,3) -- (4.5,1,3);
\draw (4,2,3) -- (4.5,2,3);
\draw (4,3,3) -- (4.5,3,3);
\draw (4,1,2) -- (4.5,1,2);
\draw (4,2,2) -- (4.5,2,2);
\draw (4,3,2) -- (4.5,3,2);
\draw (4,1,1) -- (4.5,1,1);
\draw (4,2,1) -- (4.5,2,1);
\draw (4,3,1) -- (4.5,3,1);

\draw (1,1,4) -- (1,1,4.5);
\draw (1,2,4) -- (1,2,4.5);
\draw (1,3,4) -- (1,3,4.5);
\draw (2,1,4) -- (2,1,4.5);
\draw (2,2,4) -- (2,2,4.5);
\draw (2,3,4) -- (2,3,4.5);
\draw (3,1,4) -- (3,1,4.5);
\draw (3,2,4) -- (3,2,4.5);
\draw (3,3,4) -- (3,3,4.5);

\draw (1,4,1) -- (1,4.5,1);
\draw (1,4,2) -- (1,4.5,2);
\draw (1,4,3) -- (1,4.5,3);
\draw (2,4,1) -- (2,4.5,1);
\draw (2,4,2) -- (2,4.5,2);
\draw (2,4,3) -- (2,4.5,3);
\draw (3,4,1) -- (3,4.5,1);
\draw (3,4,2) -- (3,4.5,2);
\draw (3,4,3) -- (3,4.5,3);
\end{tikzpicture}
\caption{An enclosed region and the bad edges incident to the surface.}\figlab{fig:cube}
\end{figure*}
}
\begin{abstract}
We study the problem of finding all $k$-periods of a length-$n$ string $S$, presented as a data stream. $S$ is said to have $k$-period $p$ if its prefix of length $n-p$  differs from its suffix of length $n-p$  in at most $k$ locations.

We give a one-pass streaming algorithm that computes the $k$-periods of a string $S$ using $\poly(k, \log n)$ bits of space, for $k$-periods of length at most $\frac{n}{2}$. 
We also present a two-pass streaming algorithm that computes $k$-periods of $S$ using $\poly(k, \log n)$  bits of space, regardless of period length. We complement these results with comparable lower bounds.
\end{abstract}

\section{Introduction}
In this paper we are interested in finding (possibly imperfect) periodic trends in sequences given as streams.
Informally, a sequence is said to be \emph{periodic} if it consists of repetitions of a block of characters; e.g.,
$abcabcabc$ consists of repetitions of $abc$, of length 3, and thus has period 3.
The study of periodic patterns in sequences is valuable in fields such as string algorithms, time series data mining, and computational biology. 
The question of finding the smallest period of a string is a fundamental building block for many string algorithms, especially in pattern matching, such as the classic Knuth-Morris-Pratt \cite{KnuthMP77} algorithm. 
The general technique for many pattern matching algorithms is to find the periods of prefixes of the pattern in a preprocessing stage, then use them as a guide for ruling out locations where the pattern cannot occur, thus improving
efficiency.

While finding exact periods is  fundamental to pattern matching, in real life, it is unrealistic to expect data to be perfectly periodic.
In this paper, we assume that even when there is a fixed period, data might subtly change over time.
In particular, we might see \emph{mismatches,} defined as locations in the sequence where a block is not the same as the previous block.
For instance, while $abababababab$ is perfectly periodic, $abababadadad$ contains one mismatch where $ab$ becomes (and stays) $ad$. 
This model captures periodic events that undergo permanent modifications over time (e.g., statistics that remain generally cyclic but experience infrequent permanent changes or errors). 
We consider our problem in the \emph{streaming} setting, where the input is received in a sequential manner, and is processed using sublinear space.  

Our problem generalizes exact periodicity studied in \cite{ErgunJS10}, where the authors give a one-pass, $\O{\log^2n}$-space algorithm for finding the smallest \emph{exact} period of stream $S$ of length $n$, when the period is at most $n/2$, as well as a linear space lower bound when the period is longer than $n/2$. 
They use two standard and equivalent definitions of periodicity: $S$ has period $p$ if it is of the form $B^{\ell}B'$ where $B$ is a block of length $p$ that appears $\ell\geq 1$ times in a row, and $B'$ is a prefix of $B$. 
For instance, $abcabcabcab$ has period 3 where $B = abc$, and $B'=ab$. 
Equivalently, the length $n-p$ prefix of $S$ is identical to its length $n-p$ suffix. 
These definitions imply that at most $k$ of the repeating blocks differ from the preceding ones. 
According to this definition, for instance, $abcabdabdae$ is $2$-periodic with period 3, with the mismatches occurring at positions 6 and 11.
 
In order to allow mismatches in $S$ while looking for periodicity in small space, we utilize the fingerprint data structure introduced for pattern matching with mismatches by \cite{PoratP09,CliffordFPSS16}. 
Ideally, one would hope to combine results from \cite{ErgunJS10} and \cite{CliffordFPSS16} to readily obtain an algorithm for detecting $k$-periodicity. 
Unfortunately, reasonably direct combinations of these techniques do not seem to work. 
This is due to the fact that, in the presence of mismatches, the essential structural properties
of periods break down. 
For instance, in the exact setting, if $S$ has periods $p$ and $q$, it must also have period $r$, where $r$ is any positive multiple of $p$ or $q$. 
It must also have period $d=gcd(p, q)$. 
These are not necessarily true when there are mismatches; as an example consider the following.
\begin{example}
$S=aaaaba$ has only one mismatch where $S[i]\neq S[i+2]$ (over all non range-violating values of $i$); likewise where $S[i]\neq S[i+3]$, thus $S$ is $1$-periodic with periods 2 and 3. $S$ is \emph{not} 
$1$-periodic with period $1=\gcd{2,3}$  as it has \emph{two} mismatches where 
$S[i]\neq S[i+1].$
\end{example}
In the exact setting the smallest period $t$ determines the entire structure of $S$ as all other periods must be multiples of $t$.
This property does not necessarily hold when we allow mismatches, thus the smallest period does not carry as much information as
in the exact case. 
Similarly, overlaps of a pattern with itself in $S$ exhibits a much less well-defined periodic structure in the presence of mismatches. 
This makes it much harder to achieve the fundamental space reduction achievable in exact periodicity computation, where this kind of structure is crucially exploited.

\subsection{Our Results}
Given the structural challenges introduced by the presence of mismatches, we first focus on understanding the unique structural properties of $k$-periods and the relationship between the period $p$, and  the number of mismatches $k$ (See \thmref{thm:final}).
This understanding gives us tools for ``compressing'' our data into sublinear space. We proceed to present the following
on a given stream $S$ of length $n$:
\begin{enumerate}
\item a two-pass streaming algorithm that computes all $k$-periods of $S$ using $\O{k^4\log^9 n}$ space, {\em regardless of period length} (see \secref{sec:twopass})
\item a one-pass streaming algorithm that computes all $k$-periods of length at most $n/2$ of $S$ using $\O{k^4\log^9 n}$ space (see \secref{sec:onepass})
\item  a lower bound that  any one-pass streaming algorithm that computes all $k$-periods of $S$ requires $\Omega(n)$ space (see \secref{sec:lb})
\item  a lower bound that for $k =o(\sqrt{n})$ with $k>2$, any one-pass streaming algorithm
that computes all $k$-periods of $S$ with probability at least $1-1/n$ requires
$\Omega(k \log n)$ space, even under the promise that the $k$-periods are of
length at most $n/2$. (see \secref{sec:lb})
\end{enumerate}
Given the above results, it is trivial to modify the algorithms to return, rather
than all $k$-periods, the smallest, largest, or any particular $k$-period of $S$.
\subsection{Related Work}
Our work extends two natural directions in sublinear algorithms for strings: on one hand the study of the repetitive structure of long strings, and on the other hand the notion of approximate matching of patterns, in which the algorithm can detect a pattern even when some of it got corrupted.
 
In the first line of work, Erg{\"{u}}n \etal\ \cite{ErgunJS10} initiate the study of streaming algorithms for detecting the period of a string, using $poly(\log n)$ bits of space. 
Indyk \etal\ \cite{KoudasIM00} also studied mining periodic patterns in streams,  \cite{ElfekyAE06} studied periodicity in time-series databases and online data, and Crouch and McGregor \cite{CrouchM11} study periodicity via linear sketches.
\cite{ErgunMS10} and \cite{LachishN11} studied the problem of distinguishing periodic strings from aperiodic ones in the property testing model of sublinear-time computation. 
Furthermore, \cite{amir2010approximate} studied approximate periodicity in RAM model under the Hamming and swap distance metrics. 
 
The pattern matching literature is a vast area (see \cite{ApostolicoG:1997} for a survey) with many variants.
Following the pattern matching streaming algorithm of Porat and Porat \cite{PoratP09}, Clifford \etal\ \cite{CliffordFPSS16} recently show improved streaming algorithms for the $k$-mismatch problem, as well as offline and online variants. 
We adapt the use of sketches from \cite{CliffordFPSS16} though there are some other works with different sketches for strings (\cite{andoni2013homomorphic}, \cite{clifford2009coding}, \cite{radoszewski2016streaming} and \cite{porat2007improved}). 
\cite{clifford2013space} also showed several lower bounds for online pattern matching problem. 
 
This line of work is also related to the detection of other natural patterns in strings, such as palindromes or near palindromes. 
Erg{\"{u}}n \etal\ \cite{BerenbrinkEMA14} initiate the study of this problem and give sublinear-space algorithms, while \cite{GawrychowskiMSU16} show lower bounds. 
In recent work, \cite{GrigorescuSZ17} extend this problem to finding near-palindromes (i.e., palindromes with possibly a few corrupted entries). 
 
Many ideas used in these sublinear algorithms stem from related work in the classical offline model. 
The well-known KMP algorithm \cite{knuth1977fast} initially used periodic structures to search for patterns within a text. 
Galil \etal \cite{galil1983time} later improved the space performance of this pattern matching algorithm. 
Recently, \cite{gawrychowski2013optimal} also used the properties of periodic strings for pattern matching when the strings are compressed. 
These interesting properties have allowed several algorithms to satisfy some non-trivial requirements of respective models (see \cite{golan2016esa}, \cite{CliffordFPSS15} for example). 
\section{Preliminaries}
We assume our input is a stream $S[1,\ldots, n]$ of length $|S|=n$ over some alphabet $\Sigma$. The $i\th$ character of $S$ is denoted $S[i]$, and the substring between locations $i$ and $j$ (inclusive) $S[i,j]$.  
Two strings $S,T\in\Sigma^n$ are said to have a {\it mismatch} at index $i$ if $S[i]\neq T[i]$, 
and their Hamming distance is the number of such mismatches, denoted $\HAM{S,T}=\Big|\{ i\mid S[i]\ne T[i]\}\Big|$.
We denote the concatenation of $S$ and $T$ by $S\circ T$. 
 
$S$ is said to have {\it period} $p$ if $S[x]=S[x+p]$ for all $1\le x\le n-p$;
more succinctly, if $S[1,n-p]=S[p+1,n]$.
In general, we say $S$ has $k$-period $p$ (i.e., $S$ has period $p$ with $k$ mismatches) if $S[x]=S[x+p]$ for all but at most $k$ valid indices $x$. 
Equivalently, $S$ has $k$-period $p$ if and only if $\HAM{S[1,n-p],S[p+1,n]}\le k$. 
\begin{observation}
\obslab{obs:num:words}
If $p$ is a $k$-period of $S$, then at most $k$ of the sequence of substrings $S[1,p],S[p+1,2p],S[2p+1,3p],\ldots$ can differ from the previous substring in the sequence. 
\end{observation}
When obvious from the context, given $k$-period $p$, we denote as a {\it mismatch} a position $i$ for which $S[i]\neq S[i+p]$. 
\begin{example}
The string $S=aaaaaabbccd$ has $3$-period equal to $1$, since $S[i]=S[i+1]$ for all valid locations $i$ except mismatches at $i=6,8,10$. 
On the other hand, $S=abcabcadcabc$ has $2$-period equal to $3$ since $S[i]=S[i+3]$ for all valid $i$ except mismatches $i=5,8$.
\end{example}
The following observation notes that the number of mismatches between two strings is an upper bound on the number of mismatches between their prefixes of equal length.
\begin{observation}
\obslab{obs:num:mismatch}
If $p$ is a $k$-period of $S$, then for any $x\le n-p$, the number of mismatches between $S[1,x]$ and $S[p+1,p+x]$ is at most $k$. 
\end{observation}
Given two integers $x$ and $y$, we denote their greatest common divisor by $\gcd{x,y}$.

We repeatedly use data structures and subroutines that use Karp-Rabin fingerprints. 
For more about the properties of Karp-Rabin fingerprints see \cite{KarpR87}, but for our purposes, the following suffice:
\begin{theorem}[\cite{CliffordFPSS16}]
\thmlab{thm:cfp:datastructure}
Given two strings $S$ and $T$ of length $n$, there exists a data structure that uses $\O{k\log^6 n}$ bits of space, and outputs whether $\HAM{S,T}>k$ or $\HAM{S,T}\le k$, along with the set of locations of the mismatches in the latter case.
\end{theorem}
From here, we use the term \emph{fingerprint} to refer to this data structure.

\subsection{The $k$-Mismatch Algorithm}
For our string-matching tasks, we utilize an algorithm from \cite{CliffordFPSS16},
whose parameters are given in \thmref{thm:cfp:algorithm}. For us, string matching is a tool rather than a goal; as a result, we require additional properties from the algorithm that are not obvious at first glance. In \corref{cor:cfp:algorithm} we consider these properties.
Throughout our algorithms and proofs, we frequently refer to this algorithm  as the \emph{$k$-Mismatch Algorithm}.

\begin{theorem}[\cite{CliffordFPSS16}]
\thmlab{thm:cfp:algorithm}
Given a pattern $P$ of length $\ell$, a text $T$ of length $n$ and some mismatch threshold $k$, there exists an algorithm that, with probability $1-\frac{1}{n^2}$, outputs all indices $i$ such that $\HAM{T[i,i+\ell-1],P}\le k$ using $\O{k^2\log^8 n}$ bits of space.
\end{theorem}

Whereas the pattern in the $k$-Mismatch Algorithm is given in advance and can be preprocessed before the text, in our case the pattern is a prefix of
the text, and the algorithm must return any matches of this pattern, starting possibly
at location 2, well within the original occurrence of the pattern itself. (Consider text `abcdabcdabcdabcd' and the pattern `abcdabcd,' the first six characters of the text. The first match starts at location 4, but the algorithm does not finish reading the full pattern until it has read location 6.) To eliminate a potential problem due to this requirement, we make modifications so that the algorithm can search for all matches in $S$ of a 
prefix of $S$.

\newcommand{\corcfp}{Given a string $S$ and an index $x$, there exists an algorithm which, with probability $1-\frac{1}{n^2}$, outputs all indices $i$ where $\HAM{S[1,x],S[i+1,i+x]}\le k$ using $\O{k^2\log^8 n}$ bits of space.}
\begin{corollary}\corlab{cor:cfp:algorithm}
\corcfp
\end{corollary}
\begin{proof}
We claim that the algorithm of \thmref{thm:cfp:algorithm} can be arranged and modified to output all such indices $i$.
We need to input $S[1,x]$ as the pattern and $S[2,n]$ as the text for this algorithm.

Thus, it suffices to argue that the data structure for the pattern is built in an online fashion.
That is, after reading each symbol of the pattern, the data structure corresponding to the prefix of the pattern that has already been read is updated and ready to use.
Moreover, the process of building the data structure for the text should not depend on the pattern.
The only dependency between these two processes can be that they need to use the same randomness.
Therefore, the algorithm only needs to decide the randomness before starting to process the input and share it between processes.

The algorithm of \thmref{thm:cfp:algorithm} has a few components, explained in the proof of Theorem 1.2 in \cite{CliffordFPSS16}.
Here, we go through these components and explain how they satisfy the conditions we mentioned. 

The main data structure for this algorithm is also used in \thmref{thm:cfp:datastructure}. 
In this data structure, each symbol is partitioned to various subpatterns determined by the index of the symbol along with predetermined random primes. 
Each subpattern is then fed to a dictionary matching algorithm. 
The dictionary entries are exactly the subpatterns of the original patterns and thus can be updated online. 

The algorithm also needs to consider run-length encoding for each of these subpatterns in case they are highly periodic.
It is clear that run-length encoding can be done independently for the pattern and the text.

Finally the approximation algorithm (Theorem 1.3 of \cite{CliffordFPSS16}) uses a similar data structure to \thmref{thm:cfp:datastructure}, but with different magnitudes for primes.
Thus, the entire algorithm can be modified to run in an online fashion.
\end{proof}
\section{Our Approach}
Our approach to find all the $k$-periods of $S$ is to first determine a set $\mathcal{T}$ of candidate $k$-periods, which is guaranteed to be a superset of all the true $k$-periods.  
We first describe the algorithm to find the $k$-period in two passes.
In the first pass, we let $\mathcal{T}$ be the set of indices $\pi$ that satisfy 
$$\HAM{S[1,x],S[\pi+1,\pi+x]}\le k,$$
for some appropriate value of $x$ that we specify later. 
Note that by \obsref{obs:num:mismatch}, all $k$-periods must satisfy the above inequality.
We show that even though $\mathcal{T}$ may be linear in size, we can succinctly represent $\mathcal{T}$ by adding a few additional indices into $\mathcal{T}$. 
We then show how to use the compressed version of $\mathcal{T}$ during the second
pass to verify the candidates and output the true $k$-periods of $S$.

This strategy does not work if we are allowed only one pass; by the time we discover a candidate $k$-period $p$, it may be too late for us to start collecting the extra data needed to verify $p$ (in the two-pass version this is not a problem, as the extra pass allows us to go back to the start of $S$ and any needed data).
We approach this problem by utilizing a trick from \cite{ErgunJS10} of identifying candidate periods $p$ using non-uniform criteria depending on the value of $p$. 
Using this idea, once a candidate period is found, it is not too late to verify that it is a true $k$-period, and the data can still be compressed into sublinear size.

Perhaps the biggest hidden challenge in the above approach is due to the major structural differences between exactly periodic and $k$-periodic strings;  $k$-periodic strings show much less structure than exactly periodic strings. 
As a result, incremental adaptations of existing techniques on periodic strings do not yield corresponding schemes for
$k$-periodic strings. 
In order to achieve small space, one needs to explore the weaker structural properties of
$k$-periodic streams. 
A large part of the effort in this work is in formalizing said structure (see \secref{sec:structural}), culminating
in \thmref{thm:main} and its proof, as well as exploring its application to our algorithms.

To show lower bounds for randomized algorithms finding the smallest $k$-period, we use a strategy similar to that in \cite{ErgunJS10}, using a reduction from the Augmented Index Problem.
To show lower bounds for randomized algorithms finding the smallest $k$-period given the promise that the smallest $k$-period is at most $\frac{n}{2}$, we use Yao's Principle \cite{Yao77}. 

\section{Two-Pass Algorithm to Compute $k$-Periods}\seclab{sec:twopass}
In this section, we provide a two-pass, $\O{k^4\log^9 n}$-space algorithm to output all $k$-periods of $S$. 
The general approach is to first identify a superset of the $k$-periods of $S$, based on the self-similarity of $S$, detected via the $k$-Mismatch algorithm of \cite{CliffordFPSS16} as a black box. 
Unfortunately, while this tool allows us to match parts of $S$ to each other, we get only incomplete information about possible periods, and this information is not readily stored in small space due to insufficient structure. 
We explore the structure of periods with mismatches in order to come up with a technique that massages our data into a form that can be compressed in small space, and is easily uncompressed. 
During the second pass, we go over $S$ as well as the compressed data to verify the candidate periods. 

We consider two classes of periods by their length, and run two separate algorithms in parallel. 
The first algorithm identifies all $k$-periods $p$ with $p\le\frac{n}{2}$, while the second algorithm identifies all $k$-periods $p$ with $p>\frac{n}{2}$.

\subsection{Finding small $k$-periods}
Our algorithm for finding periods of length at most $n/2$ proceeds in two passes. 
In the first pass, we identify a set $\mathcal{T}$ of candidate $k$-periods, and formulate its compressed representation, $\mathcal{T}^C$.
In the second pass, we recover each index from $\mathcal{T}^C$ and verify whether or not it is a $k$-period.
We need $\mathcal{T}$ and $\mathcal{T}^C$ to satisfy four properties.
\begin{enumerate}
\item
\label{prop:1}
All true $k$-periods (likely accompanied by some candidate $k$-periods that are false positives) are in $\mathcal{T}$.

\item
\label{prop:2}
$\mathcal{T}^C$ can be stored in sublinear space. 

\item
\label{prop:3}
$\mathcal{T}$ can be fully recovered from $\mathcal{T}^C$ in small space.

\item
\label{prop:4}
The verification process in the second pass weeds out those
candidates that are not true periods in sublinear space.
\end{enumerate}
We now describe our approach and show how it satisfies the above properties.

\subsection{Pass 1: \hyperref[prop:1]{Property 1}.} We crucially observe that any $k$-period $p$ must satisfy the requirement
$$\HAM{S[1,x],S[p+1,p+x]}\le k$$
for all $x\le n-p$, and specifically for $x=\frac{n}{2}$. 
This observation allows us to refer to 
indices as periods, as the index $p+1$ where
the requirement is satisfied corresponds to (possible) $k$-period $p$.
For the remainder of this algorithm, we set $x=\frac{n}{2}$, and
designate the indices $p+1$ that satisfy the requirement with
 $x=\frac{n}{2}$ as candidate $k$-periods; collectively these indices serve as $\mathcal{T}$.
Since satisfying this requirement is necessary but not sufficient for a candidate to be a real $k$-period, \hyperref[prop:1]{Property 1} follows.

\subsection{Pass 1: \hyperref[prop:2]{Property 2}.}
Observe that $\mathcal{T}$ could be linear in size, so we cannot store each index explicitly. We observe that if our indices followed an arithmetic progression, they
could be kept implicitly in very succinct format (as is the case where there are
no mismatches). Unfortunately, due to the
presence of mismatches in $S$, such a regular structure does not happen. However, we
show that it is still possible to implicitly add a small number of extra indices to our candidates and end up with an arithmetic series and allow for succinct representation.
Our algorithm produces several such series, and represents each one in terms of its first
index and the increment between consecutive terms, obtaining $\mathcal{T}^C$
from $\mathcal{T}$, with the details given below. 

In order to compress $\mathcal{T}$ into $\mathcal{T}^C$, we partition $[1,x]$ into the $2mk+2$ disjoint intervals $H_j=\left[\frac{jx}{2(mk+1)}+1,\frac{(j+1)x}{2(mk+1)}\right)$, where $m=\log n$. 
The goal is, possibly through the addition of extra candidates, to represent the candidates in each interval as a single arithmetic series. 
This series will be represented by its first term, as well as the increment between its consecutive terms, $\pi_j$. 
As each new candidate arrives, we update $\pi_j$ (except for the first update,
$\pi_j$ never increases, and it may shrink by an integer factor).
Throughout the process, we maintain the invariant, by updating $\pi_j$, that the arithmetic sequence represented in $H_j$ contains all candidates in $H_j$ output by the $k$-Mismatch algorithm. 
Then it is clear that $\mathcal{T}^C$ and $\{\pi_j\}$ take sublinear space, satisfying \hyperref[prop:2]{Property 2}.

\subsection{Pass 1: \hyperref[prop:3]{Property 3}.} It remains to describe how to update $\pi_j$.
The first time we see two candidates in $H_j$, we set $\pi_j$ to be the increment between the candidates (before, it is set to -1). 
Each subsequent time we see a new candidate index in the interval $H_j$, we update $\pi_j$ to be the greatest common divisor of $\pi_j$ and the increment between the candidate and the smallest index in $\mathcal{T}\cap H_j$, which is kept explicitly. 
For instance, if our first candidate index is 10, and afterwards we receive 22, 26, 32 (assume the interval ends at 35), our $\pi_j$ values over time are -1, 12, 4, 2. Ultimately, the candidates that we will be checking in Pass 2 will be 10, 12, 14, 16, 18, \ldots, 34. 
For another example, see \figref{fig:tchj}.
\figtchj
We now need to show that the above invariant is maintained throughout the algorithm. 
To do this, we show that any $k$-period $p\in H_j$ is an increment of some multiple of $\pi_j$ away from the smallest index in $\mathcal{T}\cap H_j$. 
Then, if we insert implicitly into $\mathcal{T}$ {\it all indices} in $H_j$ whose distance from the smallest index in $\mathcal{T}\cap H_j$ is a multiple of $\pi_j$, we will guarantee that any $k$-period in $H_j$ will be included in $\mathcal{T}$. 

We now show that any $k$-period $p$ is implicitly represented in, and can be recovered from $\mathcal{T}^C$ and the values $\{\pi_j\}$ at the end of the first pass.

\begin{lemma}
\lemlab{lem:recover:indices}
If $p<\frac{n}{2}$ is a $k$-period and $p\in H_j$, then $p$ can be recovered from $\mathcal{T}^C$ and $\pi_j$.
\end{lemma}
\begin{proof}
Since $p\in H_j$ is a $k$-period, then it satisfies $\HAM{S[1,n-p],S[p+1,n]}\le k$. 
More specifically, $i=p$ satisfies 
$$\HAM{S\left[1,\frac{n}{2}\right],S\left[i+1,\frac{n}{2}+i\right]}\le k$$
and will be reported by the $k$-Mismatch Algorithm. 
If there is no other index in $\mathcal{T}^C\cap H_j$, then $p$ will be inserted into $\mathcal{T}^C$ in the first pass, so $p$ can clearly be recovered from $\mathcal{T}^C$.

On the other hand, if there is another index $q$ in $\mathcal{T}^C\cap H_j$, then $\pi_j$ will be updated to be a divisor of the pairwise distances. 
Hence, the increment $p-q$ is a multiple of $\pi_j$. Any change that might later 
happen to $\pi_j$ will be due to a gcd operation, and thus, will reduce it by a factor by at least $2$. Thus, $p-q$ will remain a multiple of the final value of $\pi_j$, and $p$ will be recovered
at the end of the first pass as a member of $\mathcal{T}$.
\end{proof}
Thus \hyperref[prop:3]{Property 3} is satisfied.
The first pass algorithm in full appears below.
\begin{mdframed}
\underline{(To determine any $k$-period $p$ with $p\leq\frac{n}{2}$):}
\vskip 0.05in\noindent
First pass: 
\begin{enumerate}
\item
Initialize $\pi_j=-1$ for each $0\le j< 2k\log n+2$. 
\item
Initialize $\mathcal{T}^C=\emptyset$.
\item
For each index $i$ such that (using the $k$-Mismatch algorithm)
$$\HAM{S\left[1,\frac{n}{2}\right],S\left[i+1,\frac{n}{2}+i\right]}\le k$$
\begin{itemize}[]
\item
For the integer $j$ for which $i$ is in the interval $H_j=\left[\frac{jn}{4(k\log n+1)}+1,\frac{(j+1)n}{4(k\log n+1)}\right)$:
\begin{enumerate}
\item
If there exists no candidate $t\in\mathcal{T}^C$ in the interval $H_j$, then add $i$ to $\mathcal{T}^C$.
\item
Otherwise, let $t$ be the smallest candidate in $\mathcal{T}^C$ and either $\pi_j=-1$ or $\pi_j>0$. 
If $\pi_j=-1$, then set $\pi_j=i-t$. 
Otherwise, set $\pi_j=\gcd{\pi_j,i-t}$.  
\end{enumerate}
\end{itemize}
\end{enumerate}
\end{mdframed}

\subsection{Pass 2: \hyperref[prop:4]{Property 4}.}
Our task in the second pass is to verify whether each candidate recovered from $\mathcal{T}^C$ and $\{\pi_j\}$ is actually a $k$-period or not. 
Thus, we must simultaneously check whether $\HAM{S[1,n-p],S[p+1,n]}\le k$ for each candidate $p$, without using linear space.
Fortunately, \thmref{thm:final} states that at most $32k^2\log n+1$ unique fingerprints for substrings of length $\pi_j$ are sufficient to recover the fingerprints of both $S[1,n-p]$ and $S[p+1,n]$ for any $p\in H_j$.

Before detailing, we first state a structural property, whose proof we defer to \secref{sec:structural}. 
This property states that the greatest common divisor of the pairwise difference of any candidate $k$-periods within $H_j$ must be a $(32k^2\log n+1)$-period. 
\newcommand{\thmfinal}{
For some $0\le j<2mk+2$, let 
$$\mathcal{I}_j=\left\{i\in H_j\,\middle|\,\HAM{S[1,x],S[i+1,i+x]}\le k\right\}.$$
For any $p_1<\ldots<p_m\in\mathcal{I}$, the greatest common divisor $d$ of $p_2-p_1,p_3-p_1\ldots,p_m-p_1$ satisfies 
$$\HAM{S[1,x],S[d+1,d+x]}\le 32mk^2+1.$$
}
\begin{theorem}\thmlab{thm:final}
\thmfinal
\end{theorem}
Observe that $\pi_j$ is exactly $d$. 
Moreover, each time the value of $\pi_j$ changes, it gets divided by an integer factor at least equal to 2, ending up finally as a positive integer.
Since $\pi_j\le n$, this change can occur at most $\log n$ times, and so $m\le\log n$. 
We now show that we can verify all candidates in sublinear space.
\begin{lemma}
\lemlab{lem:recover:prints}
Let $p_i$ be a candidate $k$-period for a string $S$, with $p_1<p_2<\ldots<p_m$ all contained within $H_j$. 
Given the fingerprints of $S[1,n-p_1]$ and $S[p_1+1,n]$, we can determine whether or not $S$ has $k$-period $p_i$ for any $1\le i\le m$ by storing at most $32k^2\log n+1$ additional fingerprints.
\end{lemma}
\begin{proof}
Consider a decomposition of $S$ into substrings $w_i$ of length $p_i$, so that $S=w_1\circ w_2\circ w_3\circ\ldots$. 
Note that each index $i$ for which $w_i\neq w_{i+1}$ corresponds with at least one mismatch. 
It follows from \obsref{obs:num:words} that there exist at most $k$ indices $i$ for which $w_i\neq w_{i+1}$. 
Thus, recording the fingerprints and locations of these indices $i$ suffice to determine whether or not there are $k$ mismatches for candidate period $p_i$.

By \thmref{thm:final}, the greatest common divisor of the difference between each term in $\mathcal{I}$ is a $(32k^2\log n+1)$-period $\pi_j$. 
Thus, $S$ can be decomposed $S=v\circ v_1\circ v_2\circ v_3\circ\ldots$ so that $v$ has length $p_1$, and each substring $v_i$ has length $\pi_j$. 
It follows from \obsref{obs:num:words} that there exist at most $32k^2\log n+1$ indices $i$ for which $v_i\neq v_{i+1}$.  
Therefore, recording the fingerprints and locations of these indices $i$ allow us to recover the fingerprint of $S[1,n-p_i]$ from the fingerprint of $S[1,n-p_{i-1}]$, since $p_i-p_{i-1}$ is a multiple of $\pi_j$.  
Similarly, we can recover the fingerprint of $S[p_i+1,n]$ from the fingerprint of $S[p_{i-1}+1,n]$. 
Hence, we can confirm whether or not $p_i$ is a $k$-period.
\end{proof}
The second pass algorithm in full follows.
\begin{mdframed}
\underline{(To determine all the $k$-periods $p$ with $p\leq\frac{n}{2}$):}
\vskip 0.05in\noindent
Second pass: 
\begin{enumerate}
\item
For each $t$ such that $t\in\mathcal{T}^C$:
\begin{enumerate}
\item
Let $j$ be the integer for which $t$ is in the interval $H_j=\left[\frac{jn}{4(k\log n+1)}+1,\frac{(j+1)n}{4(k\log n+1)}\right)$
\item
If $\pi_j>0$, then record up to $32k^2\log n+1$ unique fingerprints of length $\pi_j$ and of length $t$, starting from $t$.
\item
Otherwise, record up to $32k^2\log n+1$ unique fingerprints of length $t$, starting from $t$.
\item
Check if $\HAM{S[1,n-t],S[t+1,n]}\le k$ and return $t$ if this is true.
\end{enumerate}
\item
For each $t$ which is in interval $H_j=\left[\frac{jn}{4(k\log n+1)}+1,\frac{(j+1)n}{4(k\log n+1)}\right)$ for some integer $j$:
\begin{itemize}[]
\item
If there exists an index in $\mathcal{T}^C\cap H_j$ whose distance from $t$ is a multiple of $\pi_j$, then check if $\HAM{S[1,n-t],S[t+1,n]}\le k$ and return $t$ if this is true.
\end{itemize}
\end{enumerate}
\end{mdframed}
This proves \hyperref[prop:4]{Property 4}.
Next, we show the correctness of the algorithm for small $k$-periods.
\begin{lemma}
For any $k$-period $p\le\frac{n}{2}$, the algorithm outputs $p$.
\end{lemma}
\begin{proof}
Since the intervals $\{H_j\}$ cover $\left[1,\frac{n}{2}\right]$, then $p\in H_j$ for some $j$. 
It follows from \lemref{lem:recover:indices} that after the first pass, $p$ can be recovered from $\mathcal{T}$ and $\pi_j$. 
Thus, the second pass tests whether or not $p$ is a $k$-period. 
By \lemref{lem:recover:prints}, the algorithm outputs $p$, as desired.
\end{proof}
\subsection{Finding large $k$-periods}
As in the previous discussion, we would like to pick candidate periods during our first pass. 
However, if a $k$-period $p$ satisfies $p>\frac{n}{2}$, then clearly it will no longer satisfy 
$$\HAM{S\left[1,\frac{n}{2}\right],S\left[p+1,p+\frac{n}{2}\right]}\le k,$$
as $p+\frac{n}{2}>n$, and $S\left[p+\frac{n}{2}\right]$ is undefined. 
Instead, recall that $\HAM{S[1,x]=S[p+1,p+x]}\le k$ for all $x\le n-p$. 
Ideally, when choosing candidate periods $p$ based on their
satisfying this formula, 
we would like to use as large
an $x$ as possible without exceeding $n-p$, 
but we cannot do this without knowing the value of $p$. 
Instead, \cite{ErgunJS10} observes we can try exponentially decreasing values of $x$: 
we run $\log n$ instances of the algorithm sequentially, with $x=\frac{n}{2},\frac{n}{4},\ldots$, since one of these values of $x$ must be the largest one that does not lead to an
illegal index of $S$. 
Therefore, the desired instance produces $p$, while all other instances do not. 
\begin{mdframed}
\underline{(To determine a $k$-period $p$ if $p>\frac{n}{2}$):}
\vskip 0.05in\noindent
First pass:
\begin{enumerate}
\item
Initialize $\pi^{(m)}_j=-1$ for each $0\le j<2k\log n+2$ and $0\le m\le\log n$. 
\item
Initialize $\mathcal{T}_m^C=\emptyset$.
\item
For each index $i$, let $r$ be the largest $m$ such that $\frac{n}{2}+\frac{n}{4}+\ldots+\frac{n}{2^r}\le i$. 
Using the $k$-Mismatch algorithm, check whether 
$$\HAM{S\left[1,\frac{n}{2^r}\right],S\left[i+1,i+\frac{n}{2^r}\right]}\le k.$$
If so, let $R=\frac{n}{2}+\frac{n}{4}+\ldots+\frac{n}{2^r-1}$ and $j$ be the integer for which $i$ is in the interval 
\[H^{(r)}_j=\left[R+\frac{nj}{2^{r+1}(k\log n+1)}+1,R+\frac{n(j+1)}{2^{r+1}(k\log n+1)}\right)\]
\begin{enumerate}
\item
If there exists no candidate $t\in\mathcal{T}_r^C$ in the interval $H^{(r)}_j$, then add $i$ to $\mathcal{T}_{r}^C$.
\item
Otherwise, let $t$ be the smallest candidate in $\mathcal{T}_r^C$ and either $\pi^{(r)}_j=-1$ or $\pi^{(r)}_j>0$. 
If $\pi^{(r)}_j=-1$, then set $\pi^{(r)}_j=i-t$. 
Otherwise, set $\pi^{(r)}_j=\gcd{\pi^{(r)}_j,i-t}$.  
\end{enumerate}
\end{enumerate}
\end{mdframed}
This partition of $[1,n]$ into the disjoint intervals $\left[1,\frac{n}{2}\right]$, $\left[\frac{n}{2}+1,\frac{n}{2}+\frac{n}{4}\right]$, $\ldots$ guarantees that any $k$-period $p$ is contained in one of these intervals. 
Moreover, the intervals $\{H^{(r)}_j\}$ partition 
\[\left[\frac{n}{2}+\frac{n}{4}+\ldots+\frac{n}{2^{r-1}},\frac{n}{2}+\ldots+\frac{n}{2^r}\right],\]
and so $p$ can be recovered from $\mathcal{T}_r^C$ and $\{\pi^{(r)}_j\}$.
We now present the algorithm for the second-pass to find all $k$-periods $p$ for which $p>\frac{n}{2}$. 
\begin{mdframed}
Second pass: 
\begin{enumerate}
\item
For each $t$ and any $r$ such that $t\in\mathcal{T}_r^C$:
\begin{enumerate}
\item
Let $R=\frac{n}{2}+\frac{n}{4}+\ldots+\frac{n}{2^{r-1}}$ and $j$ be the integer for which $t$ is in the interval 
\[H^{(r)}_j=\left[R+\frac{nj}{2^{r+1}(k\log n+1)}+1,R+\frac{n(j+1)}{2^{r+1}(k\log n+1)}\right)\]
\item
If $\pi^{(r)}_j>0$, then record up to $32k^2\log n+1$ unique fingerprints of length $\pi^{(r)}_j$ and of length $t$, starting from $t$.
\item
Otherwise, record up to $32k^2\log n+1$ unique fingerprints of length $t$, starting from $t$.
\item
Check if $\HAM{S[1,n-t],S[t+1,n]}\le k$ and return $t$ if this is true.
\end{enumerate}
\item
For each $t$ which is in interval $H^{(r)}_j=\left[R+\frac{nj}{2^{r+1}(k\log n+1)}+1,R+\frac{n(j+1)}{2^{r+1}(k\log n+1)}\right)$ for some integer $j$:
\begin{enumerate}
\item
If there exists an index in $\mathcal{T}_r^C\cap H^{(r)}_j$ whose distance from $t$ is a multiple of $\pi^{(r)}_j$, then check if $\HAM{S[1,n-t],S[t+1,n]}\le k$ and return $t$ if this is true.
\end{enumerate}
\end{enumerate}
\end{mdframed}
Since correctness follows from the same arguments as the case where $p\le\frac{n}{2}$, it remains to analyze the space complexity of our algorithm.
\begin{theorem}
There exists a two-pass algorithm that outputs all the $k$-periods of a given string using $\O{k^4\log^9 n}$ space.
\end{theorem}
\begin{proof}
In the first pass, for each $\mathcal{T}_m$, we maintain a $k$-Mismatch algorithm which requires $\O{k^2\log^8 n}$ bits of space, as in \corref{cor:cfp:algorithm}.
Since $1\le m\le\log n$, we require $\O{k^2\log^9 n}$ bits of space in total.
In the second pass, we keep up to $\O{k^2\log n}$ fingerprints for any set of indices in $\mathcal{T}_m$.
Each fingerprint requires space $\O{k\log^6 n}$ and there may be $\O{k\log n}$ indices in $\mathcal{T}_m$ for each $1\le m\le\log n$, for a total of $\O{k^4\log^7 n}$ bits of space. 
Thus, $\O{k^4\log^9 n}$ bits of space suffice for both passes.
\end{proof}
\section{One-Pass Algorithm to Compute $k$-Periods}\seclab{sec:onepass}
We now give a one-pass algorithm that outputs all the $k$-periods smaller than $\frac{n}{2}$. 
Similar to two-pass algorithm, we have two processes running in parallel. 
The first process handles all the $k$-periods $p$ with $p\le\frac{n}{4}$, while the second process handles the $k$-periods $p$ with $p>\frac{n}{4}$. 
Both processes are designed again based on the crucial observation that all the $k$-periods $p$ must satisfy $\HAM{S[1,x],S[p+1,p+x]}\le k$ for all $x\le n-p$.
In the first process, we set $x=\frac{n}{2}$ and find all indices $i$ such that $S\left[i+1,i+\frac{n}{2}\right]$ has at most $k$ mismatches from $S\left[1,\frac{n}{2}\right]$. 

The second process cannot use the same approach, because the $k$-Mismatch Algorithm reports that index $i$ is a candidate after reading position $\frac{n}{2}+i$, at which point we have already passed $n-i$. This means that the fingerprint of $S[1,n-i]$ cannot be built. 
For example, see \figref{fig:comparison}.
\figcomparison

Thus, for a fixed $p$ in the second process, if we set $x$ to be the largest power of two which does not exceed $n-2p$, the $k$-mismatch algorithm could report $p$. However, we cannot do this without knowing the value of $p$. 

Building off the ideas in \cite{ErgunJS10}, we run $\log n$ instances of the algorithm in parallel, with $x=1,2,4,\ldots$, then one of these values of $x$ must correspond to the instance of $k$-mismatch algorithm that recognizes $p$ and reports it for later verification.

\subsection{Finding small $k$-periods}
We consider all the $k$-periods $p$ with $p\le\frac{n}{4}$ for this subsection. 
Run the $k$-Mismatch algorithm to find
$$\mathcal{T}=\left\{i\,\middle|i\le\frac{n}{4},\HAM{S\left[1,\frac{n}{2}\right],S\left[i+1,i+\frac{n}{2}\right]}\le k\right\}.$$
Upon finding an index $i\in\mathcal{T}$, the algorithm uses the fingerprint for $S\left[i+1,i+\frac{n}{2}\right]$ to continue building $S[i+1,n]$. 
Simultaneously, it builds $S[1,n-i]$, and checks whether $\HAM{S[1,n-i],S[i+1,n]}\le k$. 
The algorithm identifies that $i\in\mathcal{T}$ upon reading character $i+\frac{n}{2}-1$. 
Since $i\le\frac{n}{4}$, then $i+\frac{n}{2}-1<\frac{3n}{4}\le n-i$. 
Thus, the algorithm can identify $i$ in time to build $S[1,n-i]$.
By \thmref{thm:final}, these entries can be computed from a sequence of compressed fingerprints. 

\subsection{Finding large $k$-periods}
Now, consider all the $k$-periods $p$ with $\frac{n}{4}<p\le\frac{n}{2}$. 
Let $I_m=\left[\frac{n}{2}-2^m+1, \frac{n}{2}-2^{m-1}\right]$ and for $1\le m\le\log n-1$, define
$$\mathcal{T}_m=\left\{i\,\middle|i\in I_m,\HAM{S[1,2^m],S[i+1,i+2^m]}\le k\right\}.$$
Let $\pi_m$ be a $k$-period of $S[1,2^m]$.
We first consider the case where $\pi_m\ge\frac{2^m}{4}$ and then the case where $\pi_m<\frac{2^m}{4}$. 

\begin{observation}
\cite{CliffordFPSS16}
\obslab{obs:apart}
If $p$ is a $k$-period for $S[1,n/2]$, then each $i$ such that 
$$\HAM{S\left[1,\frac{n}{2}\right],S\left[i+1,i+\frac{n}{2}\right]}\le\frac{k}{2}$$
must be at least $p$ symbols apart.
\end{observation} 

By \obsref{obs:apart}, if $\pi_m\ge\frac{2^m}{4}$, then $|\mathcal{T}_m|\le 4$. 
Moreover, we can detect whether $i\in\mathcal{T}_m$ by index $\frac{n}{2}-2^{m-1}+2^m$. 
On the other hand, $n-i\ge\frac{n}{2}+2^m+1$, and so we can properly build $S[1,n-i]$.

Now, suppose $\pi_m<\frac{2^m}{4}$. 
Since $\mathcal{T}_m$ may be linear in size, we use the same trick to obtain a succinct representation, whose properties satisfy those in \secref{sec:twopass}, while including a few additional indices. 
Let $S[2^m+1,2^{m+1}]=w_1w_2\ldots w_tw'$, where each $w_i$ has length $\pi_m$ and for $0\le d\le 3k$, let $x_d$ be the largest index such that $S[1,2^m]\circ w_1\circ w_2\circ\cdots\circ w_x$ has $d$-period $\pi_m$. 
 
Let $\mathcal{T}_m=i_1,i_2,\ldots,i_r$ in increasing order. 
Let $S\left[i_r+2^m+1,\frac{n}{2}+2^m\right]=v_1v_2\ldots v_sv'$, where each $v_i$ has length $\pi_m$ and let $y$ be the largest index such that $S[i_r+1,i_r+2^m]\circ v_1\circ v_2\circ\cdots\circ v_y$ has $3k$-period $\pi_m$. 

If $y=s$, then at most $k$ of the substrings $v_i$ can be unique by \obsref{obs:num:words}.
Moreover, by storing the fingerprints and positions of $\O{k^2\log n}$ substrings, as well as $v'$, we can recover the fingerprint of each $S[n-i_{j+1},n-i_j]$ by \lemref{lem:recover:prints}.
Thus, we keep the fingerprint of $S\left[\frac{n}{2}+1,n-i_r\right]$, and can construct the fingerprint of each $S\left[\frac{n}{2}+1,n-i_j\right]$

On the other hand if $y\neq s$, then for each $i_j$, let $\Delta$ be the number of indices $z$ such that $i_j\le z\le i_r$ and $S[z]\neq S[z+\pi_m]$. 
That is, $\Delta=|\{z|i_j\le z\le i_r, S[z]\neq S[z+\pi_m]\}|$.
Since $\pi_m$ is a $k$-period of $S[1,2^m]$, $\HAM{S[1,2^m],S[i_j+1,i_j+2^m]}\le k$, and each mismatch between $S[1,2^m]$ and $ S[i_j+1,i_j+2^m]$ can cause up to two indices $z$ such that $S[z]\neq S[z+\pi_m]$, then it follows that $0\le\Delta\le 3k$. 
Then if $y+|r-j|\neq x_{3k-\Delta}$, then $i_j\notin\mathcal{T}_m$, since $x_{3k-\Delta}$ is the largest index with $(3k-\Delta)$-period $\pi_m$, while $y$ is the largest index with $3k$-period $\pi_m$. 

Thus, for each $0\le\Delta\le 2k$, there is at most one index $j$ with $y+|r-j|\neq x_{2k+\Delta}$.
Again by \lemref{lem:recover:prints}, we can compute the fingerprint of $S\left[\frac{n}{2}+1,n-i_j\right]$ by storing the fingerprints and positions of $\O{k^2\log n}$ substrings.

Computing each $x_d$ requires determining $\pi_m$ and the fingerprint of $S[2^m-\pi_m+1,2^m]$.
Since $\pi_m\le\frac{2^m}{4}$, the algorithm determines $\pi_m$ by position $\pi_m+2^m<2^m-\pi_m+1$. 
Thus, the algorithm knows $\pi_m$ in time to start creating the fingerprint of $S[2^m-\pi_m+1,2^m]$. 

To compute $y$, we compute the fingerprint of $S[i_r+1,i_r+\pi_m]$. 
We then compute the fingerprint of each non-overlapping substring of length $\pi_m$ starting from $i_r+\pi_m$, and compare the fingerprint to the previous fingerprint. 
We only record the fingerprint of the most recent substring, but keep a running count of the number of mismatches.

\begin{theorem}
There exists a one-pass algorithm that outputs all the $k$-periods $p$ of a given string with $p\leq \frac{n}{2}$, and uses $\O{k^4\log^9 n}$ bits of space.
\end{theorem}
\begin{proof}
The process for small $k$-periods uses $\O{k^2\log^8 n}$ bits of space determining $\mathcal{T}$. 
Verifying whether an index in $\mathcal{T}$ is actually a $k$-period requires the fingerprints of $\O{k^2\log n}$ substrings, each using $\O{k\log^6 n}$ bits of space (\thmref{thm:cfp:datastructure}). This adds up to a total of $\O{k^3\log^7 n}$ bits of space.

The process for large $k$-periods has $\log n$ parallel instances of the $k$-Mismatch algorithm to compute $\mathcal{T}_m$ for $1\le m\le\log n$, using $\O{k^2\log^9 n}$ bits of space. 
To reconstruct the fingerprint of $S[1,n-i]$ for each $i\in\mathcal{T}_m$ the algorithm needs to store the fingerprints of at most $\O{k^2\log n}$ unique substrings (\lemref{lem:recover:prints}). 
Each fingerprint uses $\O{k\log^6 n}$ bits of space (\thmref{thm:cfp:datastructure}) and there can be up to $\O{k\log n}$ indices in $\mathcal{T}_m$. 
This adds up to a total of $\O{k^4\log^9 n}$ bits of space.

Thus, $\O{k^4\log^9 n}$ bits of space suffice for both processes.
\end{proof}

\section{Structural Properties of $k$-Periodic Strings}
\seclab{sec:structural}
In this section, we detail the necessary steps in proving \thmref{thm:final}.
\newline\noindent
\begin{remindertheorem}{\thmref{thm:final}}
\thmfinal
\end{remindertheorem}

We first show \thmref{thm:main}, which assumes there are only two candidate $k$-periods and both are small. 
We then relax these conditions and prove \thmref{thm:general:gcd}, which does not restrict the number of candidate $k$-periods, but still assumes that their  magnitudes are small. 
 \thmref{thm:final} considers all candidate $k$-periods in some interval.  We use the fact that the \emph{difference} between these candidates is small, thus meeting the conditions of \thmref{thm:general:gcd}, although with an increase in the number of mismatches.

To show that the greatest common divisor $d$ of any two reasonably small candidates $p<q$ for $k$-periods is also a $(16k^2+1)$-period (\thmref{thm:main}), we consider the cases where either all candidates are less than $(2k+1)d$ (\lemref{lem:mainSmall}) or some candidate is at least $(2k+1)d$ (\lemref{lem:mainBig}). 

In the first case, where all candidate period are less than $(2k+1)d$, we partition the string into disjoint intervals of a certain length, followed by partitioning the intervals further into congruence classes. 
We show in \lemref{lem:hop} that any partition which contains an index $i$ such that $S[i]\neq S[i+d]$ must also contain an index $j$ which is a mismatch from some symbol $p$ or $q$ distance away. 
Since there are at most $2k$ indices $j$, we can then bound the number of such partitions, and then extract an upper bound on the number of such indices $i$.

In the second case, where some candidate is at least $(2k+1)d$, our argument relies on forming a grid (such as in \figref{fig:lattice}) where adjacent points are indices which either differ by $p$ or $q$. 
We include $2k+1$ rows and columns in this grid.
Since $\frac{q}{d}\ge 2k+1$, then no index in $S$ is represented by multiple points in the grid.
We call an edge between adjacent points ``bad'' if the two corresponding indices form a mismatch.

\begin{observation}
\obslab{obs:path:mismatches}
$S[i]\neq S[i+d]$ only if each path between $i$ and $i+d$ contains a bad edge. 
\end{observation}
Our grid contains at most $2k$ bad edges, since $p$ and $q$ are both $k$-periods, and each index is represented at most once. 
We then show that for all but at most $(16k^2+1)$ indices $i$, there exists a path between indices $i$ and $i+d$ that avoids bad edges.
Therefore, there are at most $(16k^2+1)$ indices $i$ such that $S[i]\neq S[i+d]$, which shows that $d$ is an $(16k^2+1)$-period.

Before proving \lemref{lem:mainSmall}, we first show  that given integers $i,p,q$, we can repeatedly hop by distance $p$ or $q$, starting from $i$, ending at $i+\gcd{p,q}$, all the while staying in a ``small'' interval. 
\begin{lemma}
\lemlab{lem:hop}
Suppose $p<q$ are two positive integers with $\gcd{p,q}=d$. 
Let $i$ be an integer such that $1\le i\le p+q-d$. 
Then there exists a sequence of integers $i=t_0,\ldots,t_m=i+d$ where $|t_i-t_{i+1}|$ is either $p$ or $q$, and $1\le t_i<p+q$. 
Furthermore, each integer is congruent to $i\pmod{d}$. 
In other words, any interval of length $p+q$ which contains indices $i,i+d$ such that $S[i]\neq S[i+d]$ also contains an index $j$ such that either $S[j]\neq S[j+p]$ or $S[j]\neq S[j+q]$. 
\end{lemma}
\begin{proof}
Since $d$ is the greatest common divisor of $p$ and $q$, then there exist integers $a,b$ such that $ap+bq=d$. 
Suppose $a>0$. 
Then consider the sequence $t_i=t_{i-1}+p$ if $1\le t_{i-1}\le q$. 
Otherwise, if $t_{i-1}>q$, let $t_i=t_{i-1}-q$. 
Then clearly, each $|t_i-t_{i+1}|$ is either $p$ or $q$, and $1\le t_i<p+q$. 
That is, each $t_i$ either increases the coefficient of $p$ by one, or decreases the coefficient of $q$ by one. 
Thus, at the last time the coefficient of $p$ is $a$, $t_i=ap+bq=d$, since any other coefficient of $q$ would cause either $t_i>q$ or $t_i<1$. 
Hence, terminating the sequence at this step produces the desired output, and a similar argument follows if $b>0$ instead of $a>0$. 
Since $p\equiv q\equiv 0\pmod{d}$, then all integers in these sequence are congruent to $i\pmod{d}$.
\end{proof}
We now prove that the greatest common divisor $d$ of any two reasonably small candidates $p,q$ for $k$-periods is also a $(16k^2+1)$-period.
\begin{theorem}
\thmlab{thm:main}
For any $1\leq x \leq \frac{n}{2}$, let $\mathcal{I}=\left\{i\,\middle|i\le\frac{x}{4k+2},\HAM{S[1,x],S[i+1,i+x]}\le k\right\}$. 
For any two $p,q\in\mathcal{I}$ with $p<q$, their greatest common divisor, $d=\gcd{p,q}$ satisfies
\[\HAM{S[1,x],S[d+1,d+x]}\le (16k^2+1).\]
\end{theorem}
We now proceed to the proof of \thmref{thm:main} for the case $q<(2k+1)d$. 
\begin{lemma}\lemlab{lem:mainSmall}
\thmref{thm:main} holds when $q<(2k+1)d$.
\end{lemma}
\begin{proof}
If $x\le16k^2$, then clearly there are at most $16k^2$ indices $i$ such that $S[i]\neq S[i+d]$, and so $d$ is a $(16k^2+1)$-period. 
Otherwise, suppose $x>16k^2+1$, and by way of contradiction, that there are at least $16k^2+1$ indices $i$ such that $S[i]\neq S[i+d]$.

Consider the following two classes of intervals of length $\frac{p+q}{2}$: $\mathcal{I}_1=\left[1,\frac{p+q}{2}\right]$, $\left[p+q+1,\frac{3(p+q)}{2}\right]$, $\left[2(p+q)+1,\frac{5(p+q)}{2}\right]$, $\ldots$ and $\mathcal{I}_2=\left[\frac{p+q}{2}+1,p+q\right]$,$\left[\frac{3(p+q)}{2}+1,2(p+q)\right]$, $\left[\frac{5(p+q)}{2}+1,3(p+q)\right]$, $\ldots$. 
If there are at least $16k^2+1$ indices $i$ such that $S[i]\neq S[i+d]$, then either $\mathcal{I}_1$ or $\mathcal{I}_2$ contains at least $8k^2+1$ of these indices. 

Suppose $\mathcal{I}_1$ has at least $8k^2+1$ indices $i$ such that $S[i]\neq S[i+d]$. 
Now, consider the disjoint intervals of length $p+q$: $[1,p+q]$, $[p+q+1,2(p+q)]$, $[2(p+q)+1,3(p+q)]$, $\ldots$. 
Furthermore, for each of these intervals, consider the congruence classes modulo $d$. 
Since $x>16k^2+1$ and each of these congruence classes within an intervals have $\frac{p+q}{d}<\frac{2q}{d}\le2(2k)=4k$ indices, then $S[1,x]$ certainly contains at least $2k+1$ of these congruence classes.

If $\mathcal{I}_1$ has at least $8k^2+1$ indices $i$ such that $S[i]\neq S[i+d]$ and each congruence class within an interval contains less than $4k$ indices, then there are at least $2k+1$ congruence classes containing such an index $i$. 
Because each of these indices occur within $\mathcal{I}_1$, it follows that both $i$ and $i+d$ are contained within the interval (and therefore, the same congruence class). 
By \lemref{lem:hop}, each congruence class within an interval containing indices $i$ and $i+d$ $S[i]\neq S[i+d]$ also contains an index $j$ such that either
$S[j]\neq S[j+p]$ or $S[j]\neq S[j+q]$. 
Since there are at least $2k + 1$ congruence classes within intervals, then there are at least $2k+1$ such indices $j$. 
This either contradicts that there are at most $k$ indices $j$ such that $S[j]\neq S[j+p]$ or that there are at most $k$ indices $j$ such that $S[j] \neq S[j+q]$.

The proof for the case where $\mathcal{I}_2$ has at least $8k^2+1$ indices $i$ such that $S[i]\neq S[i+d]$ is symmetric.
\end{proof}
The following lemma considers the case where at least one of candidate periods $p$ or $q$ is at least $(2k+1)d$. 
Without loss of generality, assume $q\ge(2k+1)d$. 
We form a grid, such as in \figref{fig:lattice}, where adjacent points in the grid correspond to indices which either differ by $p$ or $q$. 
An edge between adjacent points is ``bad'' if the two corresponding indices form a mismatch. 
Otherwise, we call an edge an ``good''.

From \obsref{obs:path:mismatches}, $S[i]\neq S[i+d]$ only if each path between $i$ and $i+d$ contains a bad edge. 
Thus, if $S[i]\neq S[i+d]$, then the point in the grid corresponding to $i$ must be contained in some region whose boundary is formed by bad edges. 
We partition the indices into congruence classes modulo $d$, count the number of mismatches in each class, and aggregate the results.

That is, in a particular congruence class, we assume $p$ is a $k_1$-period, and $q$ is a $k_2$-period, where $k_1,k_2\le k$. 
Then the grid contains at most $k_1+k_2$ bad edges, which bounds the perimeter of the regions. 
From this, we deduce a generous bound of $(16k_1k_2+1)$ on the number of points inside these regions, which is equivalent to the number of indices $i$ such that $S[i]\neq S[i+d]$ in the congruence class.
We then aggregate over all congruence classes to show that $d$ is a $(16k^2+1)$-period.

\begin{lemma}
\lemlab{lem:mainBig}
Let $p\le q$ and $k$ be positive integers with $q\ge(2k+1)d$ and let $d=\gcd{p,q}$.  
Given a string $S$ and an integer $0\le m<d$, let there be $k_1>0$ indices $i\equiv m\pmod{d}$ such that $S[i]\neq S[i+p]$ and $k_2>0$ indices $i\equiv m\pmod{d}$, not necessarily disjoint, such that $S[i]\neq S[i+q]$ and $k_1,k_2\le k$. 
If $d=\gcd{p,q}$, then there exist at most $k_1k_2$ indices $i\equiv m\pmod{d}$ such that $S[i]\neq S[i+d]$.
\end{lemma}
\begin{proof}
Consider a pair of indices $(i,i+d)$ with $S[i]\neq S[i+d]$ in congruence class $m\pmod{d}$. 
We ultimately want to build a grid of ``large'' size around $i$, but this may result in illegal indices if $i$ is too small or too large. 
Therefore, we first consider the case where $k(p+q)\le i\le x-k(p+q)$, where we can place $i$ in the center of the grid. 
We then describe a similar argument with modifications for $i<k(p+q)$ or $i>x-k(p+q)$, when we must place $i$ near the periphery of the grid.

Given index $i$ with $k(p+q)\le i\le x-k(p+q)$, we define a $(2k+1)$-{\em grid centered at $i$} on a subset of indices of $S[1,x]$ as follows: the node at the center of the grid is $i$ and for any node $j$, the nodes $j+p$, $j+q$, $j-p$ and $j-q$ are the top, right, bottom and left neighbors of $j$, respectively. 
We include $(2k+1)$ rows and columns in this grid, so that $i$ is the intersection of the middle row and the middle column. 
See \figref{fig:lattice} for example of such a grid. 
\figlattice
Note that since $k(p+q)\le i\le x-k(p+q)$, all points in the grid correspond to indices of $S$. 
\begin{claim}
\claimlab{claim:unique}
The points in a $(2k+1)$-grid centered at $i$ correspond to distinct indeces in $S$.
\end{claim}
\begin{proof}
Suppose, by way of contradiction, there exists some index $j$ which is represented by multiple points in the grid. 
That is, $j=i+a_1p+b_1q=i+a_2p+b_2q$ with $a_1\neq a_2$. 
Since $d=\gcd{p,q}$, there exist integers $r,s$ with $p=rd$, $q=sd$, and $\gcd{r,s}=1$. 
Then $(a_1-a_2)p=(b_2-b_1)q$ so $(a_1-a_2)r=(b_2-b_1)s$. 
Because $\gcd{r,s}=1$, it follows that $(a_1-a_2)$ is divisible by $s=\frac{q}{d}\ge2k+1$. 
Therefore, $|a_1-a_2|\ge 2k+1$, and so $a_1$ and $a_2$ are at least $2k+1$ columns apart. 
However, this contradicts both points being in the grid, since the grid contains exactly $2k+1$ columns.
\end{proof}
\begin{claim}
\claimlab{claim:few:bad}
In each $(2k+1)$-grid there exist at least $k+1$ rows and $k+1$ columns in the grid that do not contain any bad edge.
\end{claim}
\begin{proof}
Since $\HAM{S[1,x],S[\alpha+1,\alpha+x]}\leq k$, for $\alpha=p,q$, there are at most $k$ indices $i$ for which $S[i]\neq S[i+p]$ or $S[i]\neq S[i+q]$. 
By \claimref{claim:unique}, each index is represented at most once.
Hence, there are at most $k$ vertical bad edges and at most $k$ horizontal bad edges in this grid.
Because the grid contains $2k+1$ rows and columns, then there exist at least $k+1$ rows and columns in the grid that do not contain any bad edge.
\end{proof}
We say that a row with no bad edges in a  $(2k+1)$-grid is a  \emph{no-change row}. We define a \emph{no-change column} similarly. 
\begin{claim}
\claimlab{claim:bad:path}
Suppose the following hold:
\begin{enumerate}
\item
There exists a path avoiding bad edges between $i$ and a no-change row or column in a $(2k+1)$-grid containing $i$.
\item
There exists a path avoiding bad edges between $i+d$ and a no-change row or column in a $(2k+1)$-grid containing $i+d$.
\end{enumerate}
Then there exists a path between $i$ and $i+d$ avoiding bad edges.
\end{claim}
\begin{proof}
Consider the two $(2k+1)$-grids centered at $i$ and $i+q$. 
By \claimref{claim:few:bad}, there are at least $k+1$ no-change rows in each grid, but the two grids overlap in $2k+1$ rows. 
Thus, some no-change row in the grid centered at $i$ must also be a no-change row in the grid centered at $i+q$. 
Similarly, some no-change column in the grid centered at $i$ must also be a no-change column in the grid centered at $i+p$. 
These common no-change rows and columns allow traversal between grids, as we can freely traverse between any no-change rows and columns while avoiding bad edges. 
Thus, if we can traverse from $i$ to any no-change row in the first grid, we can ultimately reach any no-change row in the final grid containing $i+d$ while avoiding all bad edges. 
Finally, if we can traverse between $i+d$ and any no-change row in the final grid, then there exists a path between $i$ and $i+d$ without any bad edges.
\end{proof}

By the contrapositive of \claimref{claim:bad:path}, it follows that if $S[i]\neq S[i+d]$,  then either the $(2k+1)$-grid centered at $i$ or the $(2k+1)$-grid centered at $i+d$, has no path that avoids bad edges from the center of the grid to a no-change row or no-change column. 
Suppose without loss of generality that all paths from $i$ to a no-change row/column within the $(2k+1)$-grid centered at $i$ has some bad edge. 
Define an {\em enclosed region} containing $i$ within the $(2k+1)$-grid centered at $i$ to be the set of points reachable from $i$ on paths containing only good edges. See \figref{fig:lattice} for an example.

Thus, to bound the number of indices $i$ such that $S[i]\neq S[i+d]$, it suffices to bound the number of points enclosed in such regions, which are themselves all contained within $(2k+1)$-grids.

We next argue that the number of unique indices that can be enclosed with $k_1$ vertical edges and $k_2$ horizontal edges is at most $\frac{k_1k_2}{2}$, even on an extended grid with no boundaries and multiple vertices which correspond to the same index.  
\begin{lemma}
\lemlab{lem:isolated}
The number of mismatches $(i,i+d)$, for $k(p+q)\le i\le x-k(p+q)$, is at most $\frac{k_1k_2}{2}$.
\end{lemma}

\begin{proof}
The proof follows from the following observations.
\begin{observation}
Two sets of indices represented by the grid points in two enclosed regions are either identical or completely disjoint.
\end{observation}
\noindent
Henceforth, we consider only one representative for each enclosed region.  
\begin{observation}
\obslab{obs:rectangle}
If two enclosed regions consist of sets of grid points $I_1$ and $I_2$, respectively, then there is a way to enclose at least $|I_1|+|I_2|$ many points in the grid using at most the same number of edges as the number of edges bounding the two enclosed regions, but such that these edges form only one enclosed region. Moreover, the new enclosed region can be made convex, and in particular a rectangle.
\end{observation}
Specifically, the rectangle contains at most $k_1$ vertical bad edges and at most $k_2$ horizontal bad edges. 
Because the total area (defined as the number of grid points) of regions enclosed by at most $k_1$ vertical bad edges and at most $k_2$ horizontal bad edges is at most $\frac{k_1k_2}{4}$, then the number of enclosed nodes cannot exceed $\frac{k_1k_2}{4}$. 
Therefore, the number of $(i,i+d)$ mismatches is at most double the number of enclosed nodes (if $i$ is enclosed, both $(i,i+d)$ and $(i-d,i)$ may be mismatches), and the proof is complete for the case $k(p+q)\le i\le x-k(p+q)$. 
Refer to \figref{fig:lattice} for an example.
\end{proof}
We now describe the cases where $i<k(p+q)$ and $i>x-k(p+q)$. 
The problem with the above grid for these values of $i$ is that many points in the grid either have value less than $0$ or greater than $x$. 
These points correspond to illegal indices, as $S[j]$ for $j<0$ or $j>x$ is nonsensical. 
Hence, we simply change the construction so that we still use $2k+1$ rows and columns in total, but that $i$ appears in the bottom left corner for $i<k(p+q)$. 
On the other hand, if $i>x-k(p+q)$, then we construct our grid so that $i$ appears in the top right corner. 
Once again, since $q\ge(2k+1)d$ and the grid contains $2k+1$ rows and columns, then each index appears at most once inside the grid. 

In both cases, the boundary of the grid serves to help enclose an area containing $i$. 
Thus, any node can be enclosed by a combination of the boundary of the grid and a number of bad edges. 
However, the boundary of the grid can be at most half of the entire perimeter of an enclosed region. 
The remaining half of the perimeter consists of at most $k_1$ vertical bad edges and $k_2$ horizontal bad edges, and so the entire area is at most $k_1k_2$. 
Then the number of enclosed nodes is at most $k_1k_2$. 
Again, the number of $(i,i+d)$ mismatches is at most double the number of enclosed nodes:
\begin{lemma}
The number of mismatches $(i,i+d)$, for $i<k(p+q)$ or $i>x-k(p+q)$, is at most $2k_1k_2$. 
\end{lemma}
See \figref{fig:side:lattice} for example.
\figsidelattice 
\end{proof}
We now complete the proof of \thmref{thm:main} by aggregating each congruence class handled in \lemref{lem:mainBig}.
\begin{proofof}{\thmref{thm:main}}
Recall that we have two cases: $q<(2k+1)d$ and $q\ge(2k+1)d$. 
\lemref{lem:mainSmall} handles the first case. 

In the second case, we observe that the indices in each congruence class modulo $d$ do not interfere with each other. 
In other words, the indices $i$, $i+d$, $i+p$ and $i+q$ are all in the same congruence class modulo $d$, as are all nodes and edges in the grids containing $i$ and $i+d$. 
Thus, points in an enclosed region for one congruence class modulo $d$ cannot be in an enclosed region for a different congruence class modulo $d$. 
Now, for $0\le m<d$, let $k_1^{(m)}$ be the number of indices $i\equiv m\pmod{d}$ such that $S[i]\neq S[i+p]$ and let $k_2^{(m)}$ be the number of indices $i\equiv m\pmod{d}$ such that $S[i]\neq S[i+q]$. 
By \obsref{obs:rectangle}, the number of enclosed points for a congruence class modulo $d$ is at most the number of points inside a rectangle with length $k_1^{(m)}$ and width $k_2^{(m)}$. 
Since $\sum k_1^{(m)}\le k$, the sum of the lengths of the rectangles is at most $k$. 
Similarly, $\sum k_2^{(m)}\le k$ implies that the sum of the lengths of the rectangles is at most $k$. 

To aggregate all these points across all congruence classes modulo $d$, we finally observe that a possibly larger enclosed number of points can be obtained if the edges from different congruence classes are all uniquely mapped into just one congruence class, and hence we may assume without loss of generality that all bad edges occur in the same congruence class. 
The following observation essentially finishes the proof.
\begin{observation}
The total number of enclosed points across all congruence classes is at most the number of points inside a square with length and width $k$, i.e. $k^2$.
\end{observation}
It follows that the total number of indices $i$ such that $S[i]\neq S[i+d]$ is at most $k^2$, which finishes the proof of \thmref{thm:main}.
\end{proofof}
We generalize \thmref{thm:main} by showing that the greatest common divisor of any $m\ge 2$ reasonably small candidates for $k$-periods is also a $(2mk^2+1)$-period. 
We emphasize that it is sufficient for $m\le\log n$, since the greatest common divisor can change at most $\log n$ times. 
\begin{theorem}
\thmlab{thm:general:gcd}
Let $\mathcal{I}=\left\{i\,\middle|i\le\frac{x}{2(mk+1)},\HAM{S[1,x],S[i+1,i+x]}\le k\right\}$. 
For any $p_1,\ldots,p_m\in\mathcal{I}$, their greatest common divisor, $d=\gcd{p_1,\ldots,p_m}$ satisfies
\[\HAM{S[1,x],S[d+1,d+x]}\le 8mk^2+1.\]
\end{theorem}
\begin{proof}
Observe that it no longer holds that the pairwise greatest common divisor between two candidates $p_i$ and $p_j$ is $d$. 
However, it suffices to consider $\delta=\gcd{p_1,p_m}$. 
If $\frac{p_m}{\delta}<2k+1$, then the proof reduces to that of \lemref{lem:mainSmall}. 
Notably, the $2k+1$ intervals $[1,p_1+p_m]$, $[p_1+p_m+1,2(p_1+p_m)]$, $\ldots$ each consist of $\delta$ disjoint congruence classes. 
We must modify $4(2k)(2k)+1$ to $4(2k)(mk)+1$ to apply the Pigeonhole Principle with $mk$ mismatched indices instead of $2k$ mismatched indices. 
Otherwise, the proof is similar to that of \lemref{lem:mainBig}, as follows.

Whereas  for two candidate $k$-periods $p_1,p_2$ we represented the indices of the stream as points in a grid, here, for a number of $m$ candidate $k$-periods we represent the indices of the stream as points in an $m$-dimensional hypergrid.
Here too, we reduce the problem to counting points inside an enclosed region within an extended $m$-dimensional hypergrid (instead of grid). 

As before,  an enclosed region containing $i$ is the set of points on the grid reachable from $i$ on paths containing only good edges. 
A point is a {\em  boundary point} of an enclosed region if it is in the enclosed region and is incident to a bad edge (if no bad edges are incident to a point in the region, then the point must be in the {\em interior} of the region). We sometimes denote by {\em boundary edges} the bad edges adjacent to boundary points.
There may again be several disjoint regions enclosing points, but like in \obsref{obs:rectangle}, the number of points enclosed by a fixed number of edges is maximized within a continguous ``convex'' set:
\begin{observation}
\obslab{obs:hyperrectangle}
If two enclosed regions consist of sets of hypergrid points $I_1$ and $I_2$, respectively, then there is a way to enclose at least $|I_1|+|I_2|$ many points in the grid using at most the same number of  boundary edges as the number of edges bounding the two enclosed regions, but such that these edges form only one enclosed region. Moreover, the new enclosed region can be made convex, and in particular a hyperrectangle.
\end{observation}

Recall that we will use the number of points enclosed within such regions as an upper bound for the number of pairs of indices $i, i+d$ that have different values in the stream $S$, since it is necessary that such an enclosed region exists in order to cause the existence of the mismatched pair.

For the sake of building up intuition, first consider the case $m=3$. Note that there are at most $2(3k)=6k$ bad edges in total (if, as before, we may run into the boundary of the hypergrid). 
Thus, the number of boundary points of the cube that forms the enclosed region is $6k$. 
For an illustration, see \figref{fig:cube}.
\figcube
Since a cube with at most $6k$ boundary points has volume at most $k^{3/2}$, it follows that the number of enclosed points is at most $k^{3/2}\le k^2$, which is what we aimed for.

For general $m$, as before, we may assume without loss of generality that all the bad edges are in the same congruence class modulo $d$. 
Since there are $m$ candidate $k$-periods, the total number of bad edges is at most $mk$. 
Similar to \figref{fig:side:lattice}, at most another $mk$ points can be on the boundary of the worst-case hyperrectangle implied in \obsref{obs:hyperrectangle}, due to the boundary of the hypergrid, corresponding to illegal indices of $S$. 
Thus, there are at most $2mk$ boundary points of the hyperrectangle. 

In particular, we may assume that all sides have the same length, and thus the hyperrectagle  is isomorphic to the hypergrid $[\ell]^m$, for some  integer $\ell$,  such that the number of boundary points  is $2mk$. More specifically, 
 a boundary point $x$ must have some coordinate $i$ such that either $x_i=1$ or $x_i=\ell$. Therefore, there are $2m\ell^{m-1}=2mk$ points on the boundary, blocking every path from points in the interior of the enclosed region to points outside the region.
Since such a hyperrectangle encloses $\ell^m= k^{m/(m-1)}\le k^2$  many points, it follows that that  the number of indices $i$ such that $S[i]\neq S[i+d]$, is again at most $k^2$, which completes the proof of the general case.
\end{proof}
Finally, we show that the distance between any $m$ candidate $k$-periods that are reasonably close to each other must be a $(32mk^2+1)$-period. 
This relaxes the constraints of \thmref{thm:general:gcd}. 
\newline\noindent
\begin{remindertheorem}{\thmref{thm:final}}
\thmfinal
\end{remindertheorem}
\begin{proofof}{\thmref{thm:final}}
Note that $p_2-p_1,p_3-p_1,\ldots,p_m-p_1$ are in $\mathcal{I}$, where 
\[\mathcal{I}=\left\{i\,\middle|i\le\frac{x}{2(mk+1)},\HAM{S[1,x],S[i+1,i+x]}\le 2k\right\}.\]
Then by \thmref{thm:general:gcd}, their greatest common divisor $d$ satisfies 
\[\HAM{S[1,x],S[d+1,d+x]}\le 8m(2k)^2+1=32mk^2+1.\]
\end{proofof}
\section{Lower Bounds}
\seclab{sec:lb}
\subsection{Lower Bounds for General Periods}
Recall the following variant of the Augmented Indexing Problem, denoted $\IND_{n,\delta}$, where Alice is given a string $S\in\Sigma^n$. 
Bob is given an index $i\in[n]$, as well as $S[1,i-1]$, and must output $S[i]$ correctly with probability at least $1-\delta$. 

\begin{lemma}
\cite{MiltersenNSW95}
The one-way communication complexity of $\IND_{n,\delta}$ is $\Omega((1-\delta)n\log|\Sigma|)$.
\end{lemma} 

\begin{theorem}
Any one-pass streaming algorithm which computes the smallest $k$-period of an input string $S$ requires $\Omega(n)$ space.
\end{theorem}
\begin{proof}
Consider the following communication game between Alice and Bob, who are given strings $A$ and $B$ respectively. 
Both $A$ and $B$ have length $n$, and the goal is to compute the smallest $k$-period of $a\circ b$. 
Then we show that any one-way protocol which successfully computes the smallest $k$-period of $a\circ b$ requires $\Omega(n)$ communication by a reduction from the augmented indexing problem. 

Suppose Alice gets a string $S\in\{0,1\}^n$, while Bob gets an index $i\in[n-1]$ and $S[1,i-1]$. 
Let $\mathbf{u}$ be the binary negation of $S[1]$, i.e., $\mathbf{u}=1-S[1]$.
Then Alice sets $A=(S[1])^k(S[2])^k\ldots(S[n])^k$ and Bob sets $B=\mathbf{u}^{k(n-i)}\circ (S[1])^k(S[2])^k\ldots(S[i-1])^k\circ\mathbf{1}^k$ so that both $A$ and $B$ have length $kn$. 
Moreover, the smallest $k$-period of $A\circ B$ is $k(2n-i)$ if and only if $S[i]=1$. 
\end{proof}

\subsection{Lower Bounds for Small Periods}
We now show that for $k=o(\sqrt{n})$, even given the promise that the smallest $k$-period is at most $\frac{n}{2}$, any randomized algorithm which computes the smallest $k$-period with probability at least $1-\frac{1}{n}$ requires $\Omega(k\log n)$ space.
By Yao's Minimax Principle \cite{Yao77}, it suffices to show a distribution over inputs such that every deterministic algorithm using less than $\frac{k\log n}{6}$ bits of memory fails with probability at least $\frac{1}{n}$.

Define an infinite string $1^10^11^20^21^30^3\ldots$, as in \cite{GawrychowskiMSU16}, and let $\nu$ be the prefix of length $\frac{n}{4}$. 
Let $X$ be the set of binary strings of length $\frac{n}{4}$ at Hamming distance $\frac{k}{2}$ from $\nu$. 
Given $x\in X$, let $Y_x$ be the set of binary strings of length $\frac{n}{4}$ with either $\HAM{x,y}=\frac{k}{2}$ or $\HAM{x,y}=\frac{k}{2}+1$. 
We pick $(x,y)$ uniformly at random from $(X,Y_x)$.

\begin{theorem}
\thmlab{thm:lb:mem}
Given an input $x\circ y$, any deterministic algorithm $\mathcal{D}$ that uses less than $\frac{k\log n}{6}$ bits of memory cannot correctly output whether $\HAM{x,y}=\frac{k}{2}$ or $\HAM{x,y}>\frac{k}{2}$ with probability at least $1-\frac{1}{n}$, for $k=o(\sqrt{n})$.
\end{theorem}
\begin{proof}
Note that $|X|=\binom{n/4}{k/2}$.
By Stirling's approximation, $|X|\ge\left(\frac{n}{2k}\right)^{k/2}\ge\left(\frac{n}{4}\right)^{k/4}$ for $k=o(\sqrt{n})$.

Because $\mathcal{D}$ uses less than $\frac{k\log n}{6}$ bits of memory, then $\mathcal{D}$ has at most $2^{\frac{k\log n}{6}}=n^{k/6}$ unique memory configurations. 
Since $|X|\ge\left(\frac{n}{4}\right)^{k/4}$, then there are at least $\frac{1}{2}(|X|-n^{k/6})\ge\frac{|X|}{4}$ pairs $x,x'$ such that $\mathcal{D}$ has the same configuration after reading $x$ and $x'$.
We show that $\mathcal{D}$ errs on a significant fraction of these pairs $x,x'$.

Let $\mathcal{I}$ be the positions where either $x$ or $x'$ differ from $\nu$, so that $\frac{k}{2}+1\le|\mathcal{I}|\le k$.
Observe that if $\HAM{x,y}=\frac{k}{2}$, but $x$ and $y$ do not differ in any positions of $\mathcal{I}$, then $\HAM{x',y}>\frac{k}{2}$. 
Recall that $\mathcal{D}$ has the same configuration after reading $x$ and $x'$, so then $\mathcal{D}$ has the same configuration after reading $x\circ y$ and $x'\circ y$.
But since $\HAM{x,y}=\frac{k}{2}$ and $\HAM{x',y}>\frac{k}{2}$, then the output of $\mathcal{D}$ is incorrect for either $x\circ y$ or $x'\circ y$.

For each pair $(x,x')$, there are $\binom{n/4-|\mathcal{I}|}{k/2}\ge\binom{n/4-k}{k/2}$ such $y$ with $\HAM{x,y}=\frac{k}{2}$, but $x$ and $y$ do not differ in any positions of $\mathcal{I}$. 
Hence, there are $\frac{|X|}{4}\binom{n/4-k}{k/2}$ strings $S(x,y)$ for which $\mathcal{D}$ errs. 
Recall that $y$ satisfies either $\HAM{x,y}=\frac{k}{2}$ or $\HAM{x,y}=\frac{k}{2}+1$ so that there are $|X|\left(\binom{n/4}{k/2}+\binom{n/4}{k/2+1}\right)$ strings $x\circ y$ in total. 
Thus, the probability of error is at least
\begin{align*}
\frac{\frac{|X|}{4}\binom{n/4-k}{k/2}}{|X|\left(\binom{n/4}{k/2}+\binom{n/4}{k/2+1}\right)}
&=\frac{1}{4} \cdot \frac{\binom{n/4-k}{k/2}}{\binom{n/4+1}{k/2+1}}
=\frac{(k/2+1)}{4}\frac{(n/4-3k/2+1)\ldots (n/4-k)}{(n/4-k/2+1) \ldots  (n/4+1)}\\
&\ge\frac{k/2+1}{n+4}\left(\frac{n/4-3k/2+1}{n/4-k/2+1}\right)^{k/2} =\frac{k+2}{2n+8}\left(1-\frac{k}{n/4-k/2+1}\right)^{k/2} \\
&\ge\frac{k+2}{2n+8}\left(1-\frac{k^2}{n/2-k+2}\right)\ge\frac{1}{n}
\end{align*}
where the last line holds for large $n$, from Bernoulli's Inequality and $k=o(\sqrt{n})$.
\end{proof}

\begin{lemma}
\lemlab{lem:nosmall}
For $k=o(\sqrt{n})$, any $k$-period of the string $S(x,y)=x\circ y\circ x\circ x$ is at least $\frac{n}{4}$.
\end{lemma}
\begin{proof}
We show that stronger result that if $p<\frac{n}{4}$, $k>2$, and $n>4(18k+1)(18k+2)$, then $|\{z|S[z]\neq S[z+p]\}|>\sqrt{\frac{n}{8}}>k$, for $k=o(\sqrt{n})$. 

Let $T=\nu\circ\nu\circ x\circ x$ and for each $z$, consider $T[z]$ and $T[z+p]$.
For each $j>0$, some position $z+p$ in $1^{2j}0^{2j}1^{2j+1}0^{2j+1}$ in the second $\nu$ corresponds with a mismatch in $z$. 
Since $\HAM{x,\nu}=\frac{k}{2}$ and $\HAM{x,y}\le\frac{k}{2}+1$, then $\HAM{S\left[1,\frac{n}{2}\right],T\left[1,\frac{n}{2}\right]}\le\frac{3k}{2}+1$. 
Each mismatch between $S$ and $T$ can cause at most two indices $z$ for which $T[z]\neq T[z+p]$ but $S[z]=S[z+p]$. 
Thus, by setting $j=6k>2\left(\frac{3k}{2}+1\right)+2k$, we have that for $\frac{n}{4}>(12k+1)(12k+2)$, there are at least $6k$ indices $z$ for which $T[z]\neq T[z+p]$, and thus at least $2k$ indices for which $S[z]\neq S[z+p]$.
\end{proof}

\begin{corollary}
If $\HAM{x,y}=\frac{k}{2}$, then the string $S(x,y)=x\circ y\circ x\circ x$ has period $\frac{n}{4}$. 
On the other hand, if $\HAM{x,y}=\frac{k}{2}+1$, then $S(x,y)$ has period greater than $\frac{n}{4}$.
\end{corollary}

\begin{theorem}
For $k=o(\sqrt{n})$ with $k>2$, any one-pass streaming algorithm which computes the smallest $k$-period of an input string $S$ with probability at least $1-\frac{1}{n}$ requires $\Omega(k\log n)$ space, even under the promise that the $k$-period is at most $\frac{n}{2}$.
\end{theorem}
\begin{proof}
By \thmref{thm:lb:mem}, any algorithm using less than $\frac{k\log n}{6}$ bits of memory cannot distinguish between $\HAM{x,y}=\frac{k}{2}$ and $\HAM{x,y}=\frac{k}{2}+1$ with probability at least $1-1/n$. 
Thus, no algorithm can distinguish whether the period of $S(x,y)$ is $\frac{n}{4}$ with probability at least $1-1/n$ while using less than $\frac{k\log n}{6}$ bits of memory.
\end{proof}
\def\shortbib{0}
\bibliographystyle{alpha}
\bibliography{references}
\end{document}